\documentclass[10pt,reqno]{amsart}
\usepackage{natbib}
\usepackage{amsaddr}
\usepackage{etoolbox}
\patchcmd{\section}{\scshape}{\bfseries\scshape}{}{}
\makeatletter
\renewcommand{\@secnumfont}{\bfseries}
\makeatother

\usepackage[utf8]{inputenc} 
\usepackage[T1]{fontenc}    
\usepackage{hyperref}       
\usepackage[noabbrev, capitalise]{cleveref}
\usepackage{url}            
\usepackage{booktabs}       
\usepackage{amsfonts}       
\usepackage{nicefrac}       
\usepackage{microtype}      
\usepackage[a4paper,margin=3cm]{geometry}

\usepackage{amsmath}
\usepackage{color}
\usepackage{enumitem}
\usepackage{amssymb}
\usepackage{amsthm}
\usepackage{graphicx}
\usepackage{dsfont}
\usepackage{subcaption}
\usepackage{tikz}

\DeclareMathOperator{\Tr}{Tr}

\DeclareMathOperator{\Id}{Id}

\DeclareMathOperator{\Supp}{Supp}

\definecolor{NavyBlue}{rgb}{0.1,0.1,0.6}

\definecolor{OliveGreen}{rgb}{0,0.6,0}

\def\bfone{{\boldsymbol 1}}

\newcommand{\diff}{\mathrm{d}}

\newcommand{\N}{\mathbb{N}}
\newcommand{\Z}{\mathbb{Z}}

\newcommand{\R}{\mathbb{R}}

\renewcommand{\leq}{\leqslant}
\renewcommand{\geq}{\geqslant}

\renewcommand{\hat}{\widehat}

\newtheorem{remark}{Remark}

\newtheorem{proposition}{Proposition}
\newtheorem{thm}{Theorem}
\newtheorem{lem}{Lemma}
\newtheorem{coro}{Corollary}

\makeatletter
\g@addto@macro{\endabstract}{\@setabstract}
\newcommand{\authorfootnotes}{\renewcommand\thefootnote{\@fnsymbol\c@footnote}}%
\makeatother

\title[Acceleration of Gossip Algorithms]{Acceleration~of~Gossip~Algorithms through~the~Euler--Poisson--Darboux~Equation}

\begin{document}
\begin{center}
	\LARGE 
	Acceleration~of~Gossip~Algorithms through~the~Euler--Poisson--Darboux~Equation 
	\par \bigskip
	
	\normalsize
	\authorfootnotes
	Raphaël Berthier\textsuperscript{1} and
	Mufan (Bill) Li\textsuperscript{2} \par \bigskip
	
	\textsuperscript{1}EPFL\footnote{While RB is currently affiliated to EPFL, the research presented here was conducted mainly at Inria, Département d’informatique de l’ENS, PSL Research University, Paris, France.}
	\smallskip \par
	\textsuperscript{2}University of Toronto and Vector Institute 
\end{center}

\begin{abstract}

Gossip algorithms and their accelerated versions have been studied exclusively in discrete time on graphs. 
In this work, we take a different approach, and consider the scaling limit of gossip algorithms in both large graphs and large number of iterations. 
These limits lead to well-known partial differential equations (PDEs) with insightful properties. 
On lattices, we prove that the non-accelerated gossip algorithm of \citet{boyd2006randomized} converges to the heat equation, and the accelerated Jacobi polynomial iteration of \citet{berthier2020accelerated} converges to the Euler--Poisson--Darboux (EPD) equation --- a damped wave equation. 
Remarkably, with appropriate parameters, the fundamental solution of the EPD equation has the ideal gossip behaviour: a uniform density over an ellipsoid, whose radius increases at a rate proportional to~$t$ --- the fastest possible rate for locally communicating gossip algorithms. 
This is in contrast with the heat equation where the density spreads on a typical scale of~$\sqrt{t}$. 
Additionally, we provide simulations demonstrating that the gossip algorithms are accurately approximated by their limiting PDEs. 

\end{abstract}

\section{Introduction}

In computer science, the large amount of data and the large size of computation networks motivate a growing interest in distributed algorithms, see \citep{assran2020advances} for instance. Among these, decentralized algorithms---where there is no master node aggregating information and distributing tasks---are appreciated for their flexibility, their robustness to node/links failures and their scalability. In this paper, we study the averaging problem, or gossip problem, a toy problem of decentralized computing where the network aims to compute the average of real values distributed along the nodes of the network, see \citep{boyd2006randomized} for instance. The algorithms that tackle this task---called gossip algorithms---are used as primitives in more complex distributed algorithms, including distributed optimization or distributed reinforcement learning, see, e.g., \citep{dimakis2010gossip,assran2020advances,szorenyi2013gossip}.

To be specific, we give each node $v$ a real value $x_0(v)$, and gossip algorithms aim at computing the average of the values held by the nodes in a decentralized fashion. \citet{boyd2006randomized} proposed the following gossip algorithm, that we refer to as \emph{simple gossip}: at each communication round $n$, each node $v$ replaces its current value $x_n(v)$ by a weighted average of the value of its neighbors and its own current value. More specifically,
\begin{align}
\label{eq:simple-gossip-general}
&x_{n+1}(v) = \sum_{\eta:\eta\sim v} W_{v,\eta} x_{n}(\eta) \, ,
\end{align}
Here, we use $\eta \sim v$ to denote that $\eta,v$ are neighbouring nodes, and $W_{v,\eta}$ to denote the weights in the averaging operation. As the number of iterations grows, the running value $x_n(v)$ of each node can be shown to converge to the global average of the initial values. However, the convergence is notoriously slow in networks with a finite-dimensional geometry, such as grids or random geometric graphs. Heuristically, simple gossip averages locally the values in the network, but fails at spreading information quickly at a larger distance. This is analogous to the heat diffusion in a continuous media that homogenizes quickly locally but only slowly on large scales: for this reason, the slow convergence of simple gossip is sometimes called a \emph{diffusivity} problem. We make this analogy more precise later in this section. 

To reach super-diffusive rates of convergence, several accelerations of the simple gossip algorithm were proposed. We mention a few of them that are representative or related to the approach of this paper. \citet{dimakis2008geographic} use the knowledge of the position of the nodes in space in order to give inertia to the information in a specific direction. \citet{even2021continuized} mimicked the proof of the accelerated coordinate gradient descent of \citet{nesterov2012efficiency} to obtain an accelerated gossip algorithm in the form of an iteration over several variables. \citet{berthier2020accelerated} designed second-order iterations using an orthogonal polynomial point of view coupled with a Jacobi approximation of the spectral measure of the network graph: the resulting accelerated gossip algorithm is called the Jacobi polynomial iteration. \citet{sardellitti2010fast} built on the analogy of simple gossip with diffusion processes; they proposed to add the discrete equivalent of an advection term in the iteration to accelerate the homogenization. 

\begin{figure}[t]
	\begin{subfigure}{0.49\linewidth}
		\begin{center}
		
			\begin{tikzpicture}[scale = 0.06]
    \tikzstyle{point}=[draw,circle,fill];
    \tikzstyle{fleche}=[-,line width=1];
    \node[] (v-1) at (-12,0) {\huge$\cdots$};
    \node[point] (v0) at (0,0) {};
    \node[point] (v1) at (16,0) {};
    \node[point] (v2) at (32,0) {};
    \node[point] (v3) at (48,0) {};
    \node[point] (v4) at (64,0) {};
    \node[] (v5) at (77,0) {\huge$\cdots$};
    \draw[fleche] (v0) to ++(-6,0);
    \draw[fleche] (v0) to (v1);
    \draw[fleche] (v1) to (v2);
    \draw[fleche] (v2) to (v3);
    \draw[fleche] (v3) to (v4);
    \draw[fleche] (v4) to ++(6,0);
	\end{tikzpicture}
	\vspace{1cm}
	
		\end{center}
		\vspace*{-0.5cm}
		\caption{graph structure}
	\end{subfigure}
	\begin{subfigure}{0.49\linewidth}
		\begin{center}
			\includegraphics[width = \linewidth]{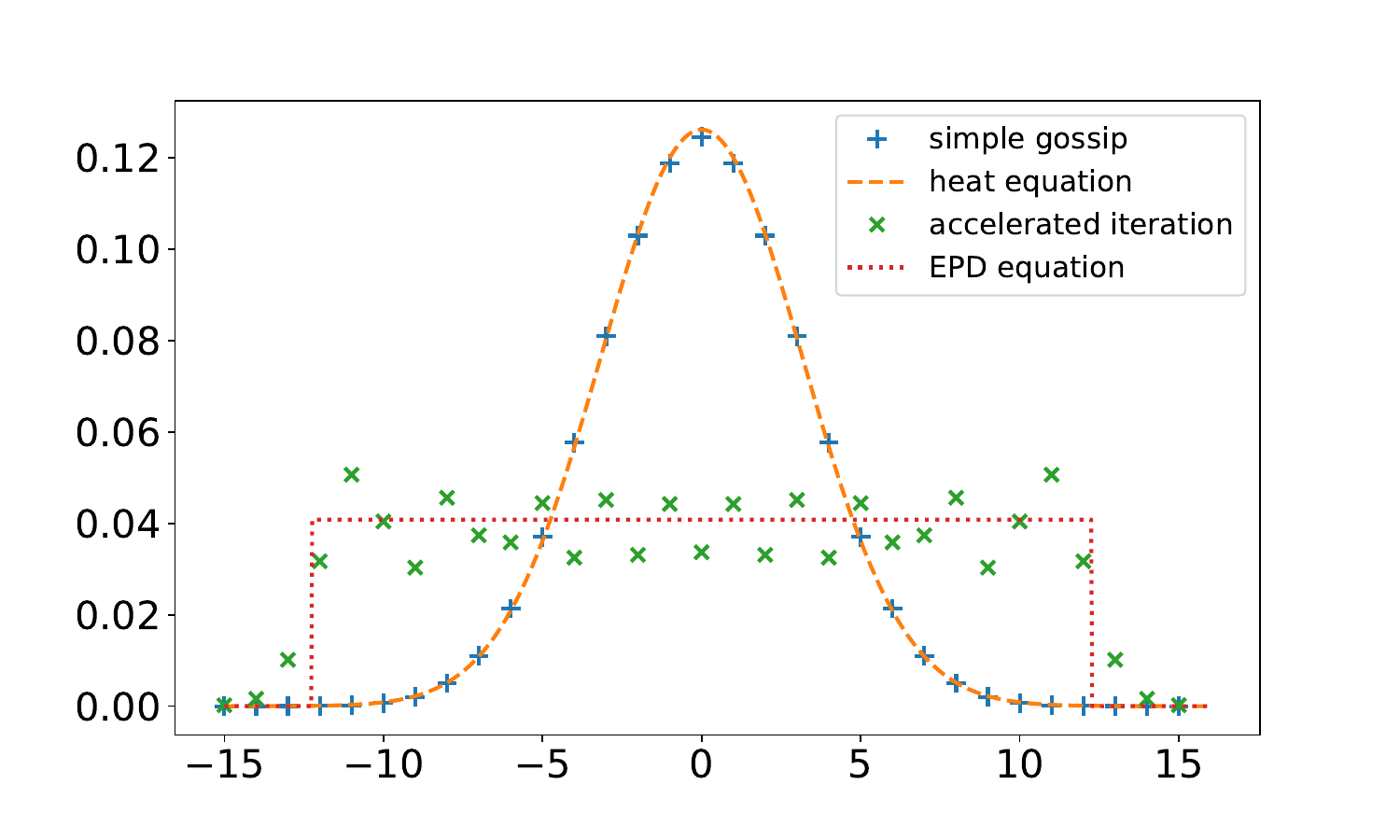}
		\end{center}
		\vspace*{-0.5cm}
		\caption{n = 15}
	\end{subfigure}
	\begin{subfigure}{0.49\linewidth}
		\begin{center}
			\includegraphics[width = \linewidth]{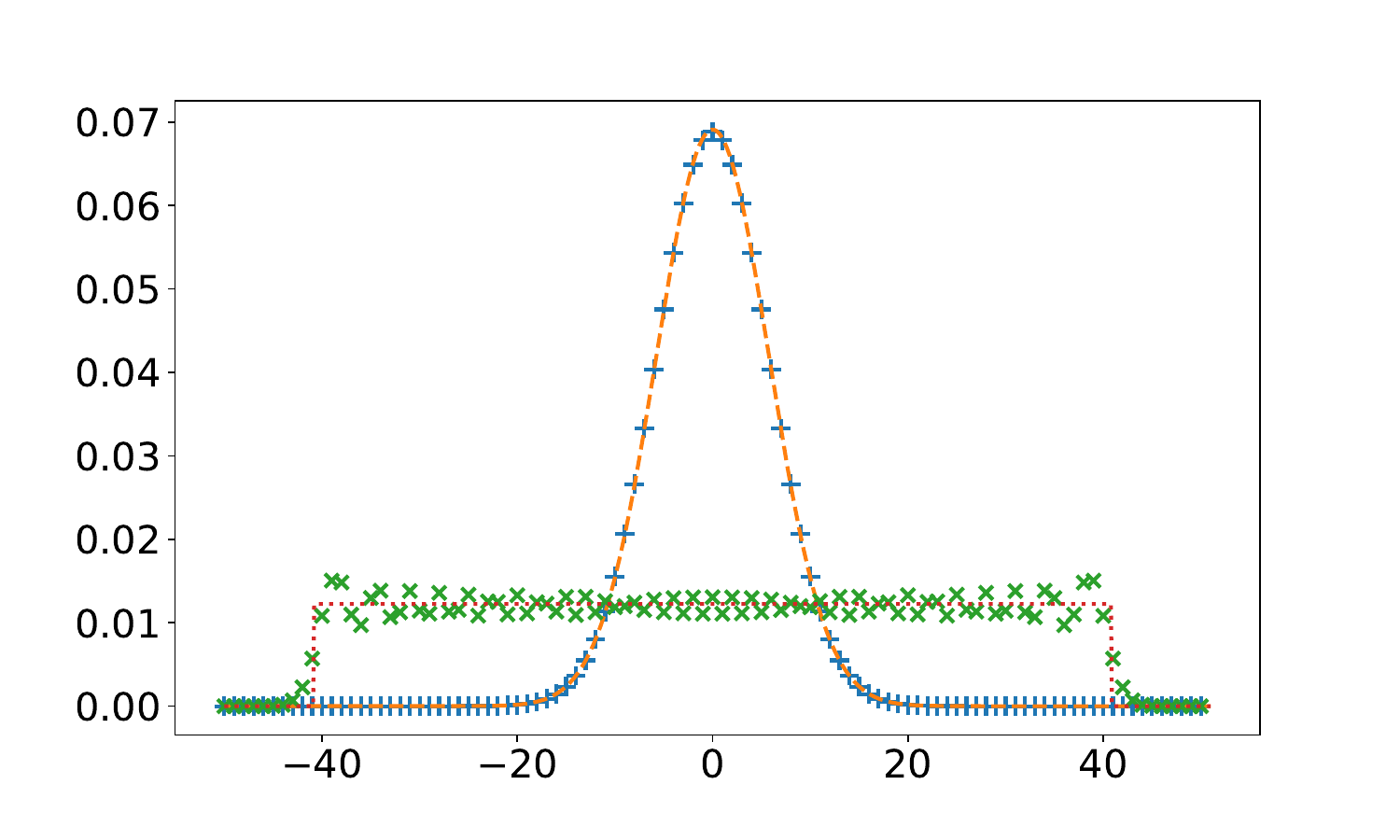}
		\end{center}
		\vspace*{-0.5cm}
		\caption{n = 50}
	\end{subfigure}
	\begin{subfigure}{0.49\linewidth}
		\begin{center}
			\includegraphics[width = \linewidth]{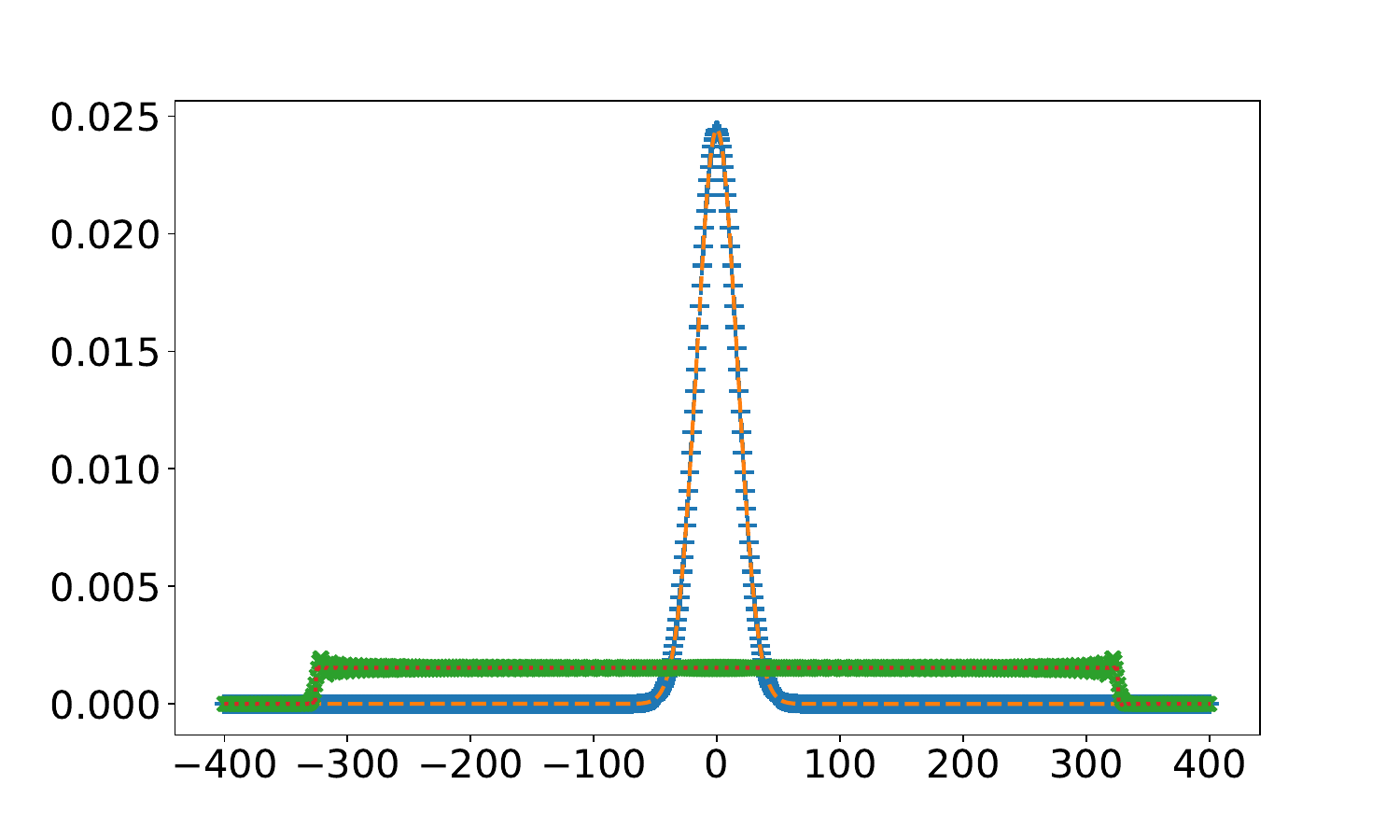}
		\end{center}
		\vspace*{-0.5cm}
		\caption{n = 200}
		\label{subfig:jpi}
	\end{subfigure}
	\caption{Comparison between gossip algorithms and their scaling limits: the simple gossip of \citep{boyd2006randomized} and the heat equation, the accelerated Jacobi polynomial iteration of \citep{berthier2020accelerated} and the Euler--Poisson--Darboux equation. All iterations were run on the line graph $\Z$ ($d=1$) and initialized from $x_0 = \bfone_0$, the vector such that $\bfone_0(0) = 1$ and $\bfone_0(v) = 0$ otherwise. We show the results $x_n(v)$ as a function of $v \in \Z$ for different numbers of iterations $n = 15,50,200$. Note that as the number of iteration increases, the description through the scaling limits improves in accuracy. The accelerated Jacobi polynomial iteration diffuses faster; it has a different scaling than the simple gossip algorithm. }
	\label{fig:scaling}
\end{figure}

\medskip\noindent
\textbf{Contributions.} We analyze and design gossip algorithms through associated partial differential equations (PDEs). We study gossip algorithms on regular lattices on $\Z^d$ (for instance, grids), and draw a rigorous connection to their corresponding PDEs via scaling limits in time and space. 

Our first contribution is to show that the scaling limit of the simple gossip algorithm on a lattice is the heat equation (see, e.g., \citep{evans1998partial})
\begin{align}
\label{eq:heat-equation}
\partial_t u  =  \frac{1}{2} \nabla_y \cdot (Q \nabla_y u)  \, , \qquad u = u(t,y) \, .
\end{align}
Here, $\nabla_y$ and $\nabla_y \cdot$ denote respectively the gradient and the divergence operator in the variable $y$. $Q$ is a $d\times d$ matrix quantifying the potential anisotropy of the diffusion: it is a function of the local averaging operation, defined more precisely in Section \ref{sec:setting}. The fundamental solution of the heat equation~\eqref{eq:heat-equation} (the weak solution when initialized at the Dirac mass $u(0,.) = \delta_0$ \citep{evans1998partial}) is a centered Gaussian density 
\begin{align}
\label{eq:solution-heat}
u(t,y) = \frac{1}{(2\pi)^{d/2}t^{d/2}(\det Q)^{1/2}} \exp\left(-\frac{1}{2t} \left\langle y, Q^{-1}y \right\rangle\right) \, .
\end{align}
The formula above shows the sub-optimality of the simple gossip method: the mass spreads on a typical scale $\Vert y \Vert \approx \sqrt{t}$, while we would like the scale to be $\Vert y \Vert \approx t$; indeed, the gossiped information can travel at most at distance $\Theta(t)$ in a time $t$ (due to the speed limit of local communications), and we would like our gossip algorithms to match this optimal speed of diffusion. Equivalently, the solution decays to $0$ at the rate $1/t^{d/2}$ in $\Vert . \Vert_\infty$, while we would like the rate to be $1/t^d$.

To this goal, we design an accelerated second-order gossip iteration: we choose the recursion coefficients so that the iteration converges to the Euler--Poisson--Darboux (EPD) equation  \citep{eulerinstitutiones,poisson1823memoire,darboux1896leccons} 
\begin{align}
\label{eq:epd}
\partial_{tt} u + \frac{d+1}{t} \partial_t u = \nabla_y \cdot \left(Q\nabla_y u\right) \, .
\end{align}
See Appendix~\ref{ap:epd} and \citet{bresters1973equation} for an introduction to the EPD equation in a more general form. The EPD equation is a wave equation with a decaying damping term. Intuitively, the wave component gives inertia to the diffusion so that the resulting PDE mixes faster, while the damping term reduces oscillation. For the precise value of the damping coefficient $\frac{d+1}{t}$ and the initial conditions $u(0,.) = \delta_0 \, ,\partial_t u(0,.) = 0$, the fundamental solution has a remarkable formula 
\begin{align}
\label{eq:fundamental-solution-d/2}
u(t,y) = \frac{\Gamma(d/2+1)}{\pi^{d/2}(\det Q)^{1/2}} \frac{1}{t^d} \bfone_{\left\{\left\langle y, Q^{-1} y \right\rangle \leq t^2 \right\}} \, .
\end{align}
This method thus has an optimal scaling: the mass spreads on a typical scale $\Vert y \Vert \approx t$ and the solution decays to $0$ at the rate $1/t^d$. Furthermore, the fundamental solution also has the perfect shape: the averaging is uniform on the ellipsoid $\left\{\left\langle y, Q^{-1} y \right\rangle \leq t^2 \right\}$. We emphasize that this solution embodies the ideal gossip behaviour: uniformly spreading information at a sharp rate (again this is the fastest possible rate for locally communicating gossip algorithms). 

Notably, the Jacobi polynomial iteration of \citet{berthier2020accelerated} is one of the possible accelerations that we identify as scaling to the EPD equation \eqref{eq:epd}. As a consequence, this paper can be seen as a more intuitive derivation of the Jacobi polynomial iteration, that was derived through algebraic methods on polynomials. 

In Figures~\ref{fig:scaling}-\ref{fig:diffusion-shape}, we provide simulations in dimension $d=1$ and $d=2$. They show that the limiting PDEs are accurate in describing the behavior of gossip algorithms as the number of iterations grows, and that the accelerated methods achieve faster diffusion.  

\begin{figure}
	\begin{subfigure}{0.48\linewidth}
		\begin{center}
			\includegraphics[width = \linewidth]{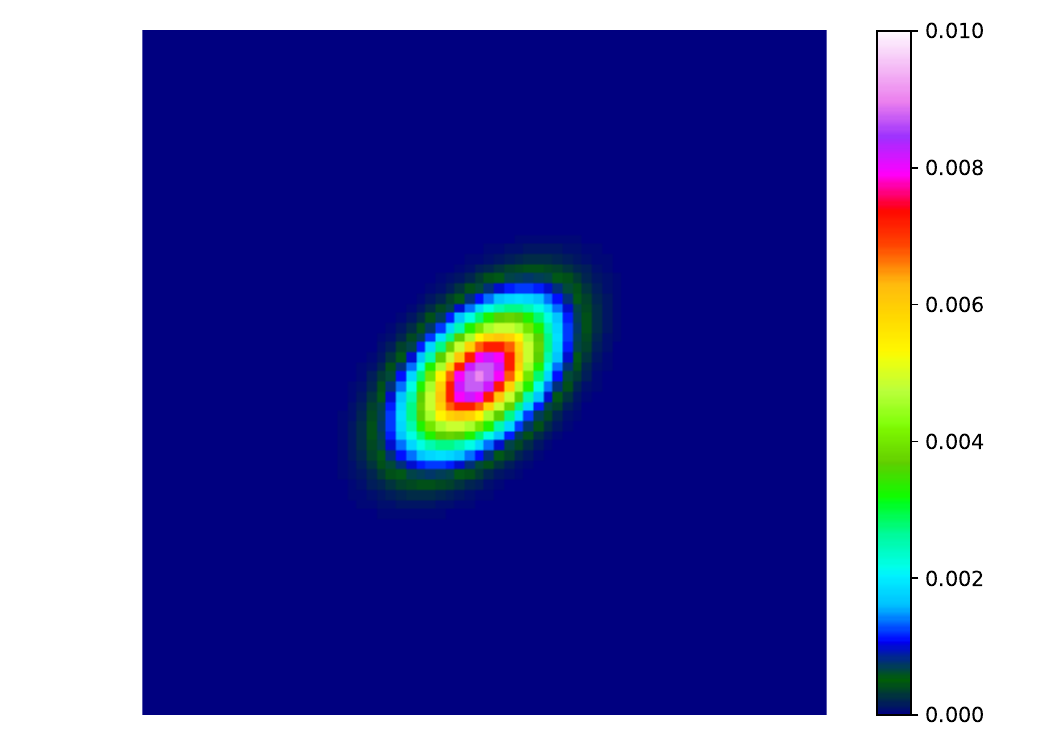}
		\end{center}
		\caption{Simple gossip}
	\end{subfigure}
	\begin{subfigure}{0.48\linewidth}
		\begin{center}
			\includegraphics[width = \linewidth]{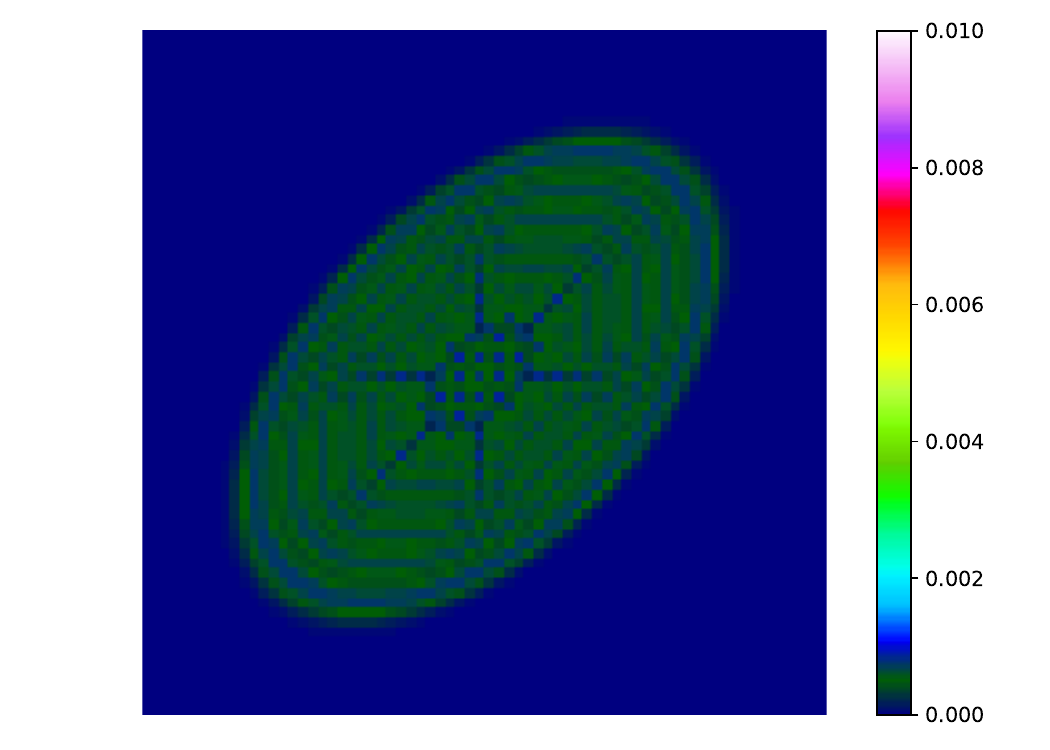}
		\end{center}
		\caption{Accelerated Jacobi polynomial iteration}
	\end{subfigure}
	\caption{Comparison between simple gossip and the accelerated Jacobi polynomial iteration on the triangular lattice (see Equation \eqref{eq:triangular} or Figure \ref{fig:drawing-graphs}(B)). We initialize from $x_0 = \bfone_0$ and we show the iterates $x_{30} = (x_{30}(v))_{v \in \Z^2}$ using a color scale. The accelerated Jacobi polynomial iteration diffuses faster than simple gossip: the mass is distributed more evenly and on a larger ellipsoid.}
	\label{fig:diffusion-shape}
\end{figure}

We intend this paper to be exploratory in the connection between gossip algorithms and PDEs, therefore our results are restricted to simplified settings.  
More precisely, our derivations and proofs assume synchronous communications---where all nodes communicate synchronously---as frequently done in the analysis of gossip algorithms, see for instance \citep{cao2006accelerated,rebeschini2017accelerated,berthier2020accelerated}. 
Moreover, the network graph must be a lattice and the gossip operation must be translation-invariant. 
Indeed, this restriction is crucial to our proof technique which relies heavily on using the Fourier transform. 
However, a Fourier approach allows us to study the fundamental solution, which cannot be easily handled by standard numerical analysis techniques \citep{genis1984finite, celia1992numerical}. 
At the same time, we believe that the insight from accelerating through the EPD equation can be used in much wider settings. 
For instance, the Jacobi polynomial iteration accelerates in many graphs with a finite-dimensional geometry, including regular grids, percolation bonds, or random geometric graphs \citep{berthier2020accelerated}. 
We expect the intuitions presented in this paper to have the same universality. 

\medskip\noindent
\textbf{Structure of the paper.} In Section \ref{sec:setting}, we set up our problem: we properly define the lattices, the gossip problem and the simple gossip algorithm. 
In Section \ref{sec:heuristic}, we give heuristic but intuitive derivations showing why simple gossip scales to the heat equation and how to build a method scaling to the EPD equation. 
Furthermore, we draw the connection with Jacobi polynomial iterations, and we discuss the open problem of extending the results to asynchronous settings or on more general graphs. 

Section \ref{sec:rigorous} provide some rigorous support to the above derivations. The simple gossip algorithm can be seen as the iteration of the law of a random walk on~$\Z^d$; the convergence to a Gaussian random variable is made rigorous by the central limit theorem and the local central limit theorem, see Section~\ref{sec:rigorous-simple}. In Section~\ref{sec:rigorous-jacobi}, we provide analog results for the convergence of the Jacobi polynomial iteration to the EPD equation~\eqref{eq:epd}: a weak limit theorem and a stronger result of local type. Finally, in Section~\ref{sec:sharp-rates}, we apply the latter result to obtain an asymptotic equivalent\footnote{We use the notation $u_n \underset{n\to\infty}{\sim} v_n$ to mean that $u_n/v_n$ converges to $1$ as $n \to \infty$.} (not only a domination) of the convergence rate for the Jacobi polynomial iteration on $\Z^d$:
\begin{align*}
\sum_{v \in \Z^d} x_n(v)^2 \underset{n\to\infty}{\sim} \frac{1}{(\det Q)^{1/2} |B(0,1)|} \frac{1}{n^d} \, , 
\end{align*}
where $|B(0,1)|$ is the volume of the unit ball in $\mathbb{R}^d$. 

\medskip\noindent
\textbf{Notation.} For $v \in \Z^d$, we denote $\bfone_v = (\bfone_v(w))_{w \in \Z^d}$ the vector with entry $\bfone_v(v)=1$ and all other entries equal to $0$. We denote $e_1, \dots, e_d$ the canonical basis of $\R^d$. $\lfloor s \rfloor$ denotes the integer part of a real number $s$. 

\section{Setting}
\label{sec:setting}

\noindent
\textbf{Lattices and translation-invariant gossip operation.} 
In this section, we introduce the notations and lattices for the gossip problem. 
In order to make the rescaling of the processes indexed by the vertices more natural, we consider here only graphs with vertex set $\Z^d$.
Instead of defining the edge sets using graph notation, we will use a more convenient (and equivalent) definition through the local averaging operation. 
Let $\omega = (\omega(v))_{v \in \Z^d}$ be a vector of non-negative reals, representing a local averaging filter on~$\Z^d$. 
Here we assume that $\omega$ has finite support and that $\sum_{v \in \Z^d} \omega(v) = 1$. 
For a vector $x = (x(v))_{v\in\Z^d}$, we define the local averaging operation by the discrete convolution on~$\Z^d$
\begin{align*}
(\omega * x)(v) = \sum_{\eta \in \Z^d} \omega(v-\eta) x(\eta) \, , 
\end{align*}
which represents a local weighted average of the neighboring values. 
For this operation to represent local communications on the graph, we would require all pairs $\{v,\eta\}$ such that $\omega(v-\eta) > 0$ to be connected in the graph. We thus take $\{\{v,\eta\} \in (\Z^d)^2 \, | \, \omega(v-\eta) > 0 \}$ to be the edge set. The support of $\omega$ represents the communication range from any vertex in $\Z^d$. 
Typically, we can take 
\begin{align}
\label{eq:standard-diffusion-vector}
\omega = \frac{1}{2d} \sum_{i = 1}^d \left(\bfone_{e_i} + \bfone_{-e_{i}}\right) \, .
\end{align}
This corresponds to allowing only nearest neighbors communicate in $\Z^d$. In this case, the underlying graph is the standard lattice on $\Z^d$. The triangular lattice in dimension $2$ can be obtained by taking 
\begin{align}
\label{eq:triangular}
\omega = \frac{1}{6}\left(\bfone_{(1,0)} + \bfone_{(-1,0)} + \bfone_{(0,1)}+ \bfone_{(0,-1)}+ \bfone_{(1,1)}+ \bfone_{(-1,-1)}\right) \, .
\end{align}
See Figure \ref{fig:drawing-graphs} for drawings of these lattices. 
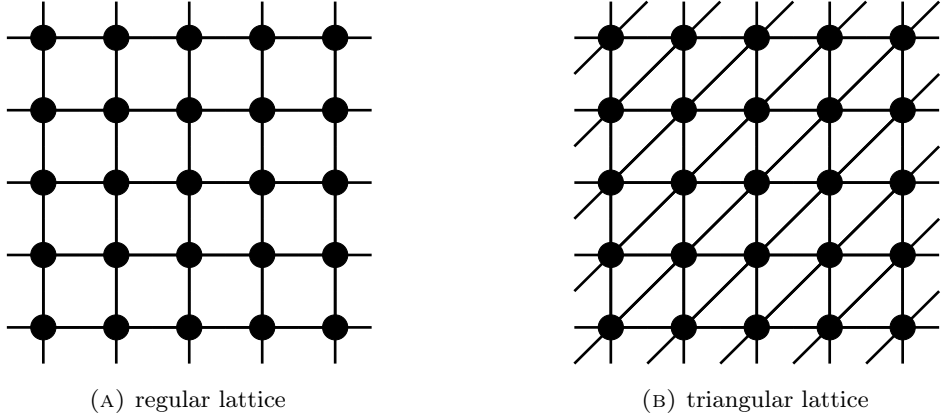
\begin{figure}
	\begin{subfigure}{0.49\linewidth}
		\begin{center}
		\begin{tikzpicture}[scale = 0.06]
    \tikzstyle{point}=[draw,circle,fill];
    \tikzstyle{fleche}=[-,line width=1];
    \foreach \i in {0,16,32,48,64}{
    \foreach \j in {0,16,32,48,64}{
    \node[point] (v0) at (\i,\j) {};
    }}
    \foreach \i in {0,16,32,48,64}{
    \foreach \j in {0,16,32,48}{
    \draw[fleche] (\i,\j) to ++(0,16);
    }}
        \foreach \i in {0,16,32,48}{
    \foreach \j in {0,16,32,48,64}{
    \draw[fleche] (\i,\j) to ++(16,0);
    }}
    \foreach \i in {0,16,32,48,64}{
    \foreach \j in {0,16,32,48,64}{
    \draw[fleche] (\i,\j) to ++(0,8);;
    \draw[fleche] (\i,\j) to ++(0,-8);;
        \draw[fleche] (\i,\j) to ++(8,0);;
    \draw[fleche] (\i,\j) to ++(-8,0);;
    }}
	\end{tikzpicture}
		\end{center}
		\caption{regular lattice}
	\end{subfigure}
	\begin{subfigure}{0.49\linewidth}
		\begin{center}
\begin{tikzpicture}[scale = 0.06]
    \tikzstyle{point}=[draw,circle,fill];
    \tikzstyle{fleche}=[-,line width=1];
    \foreach \i in {0,16,32,48,64}{
    \foreach \j in {0,16,32,48,64}{
    \node[point] (v0) at (\i,\j) {};
    }}
    \foreach \i in {0,16,32,48,64}{
    \foreach \j in {0,16,32,48}{
    \draw[fleche] (\i,\j) to ++(0,16);
    }}
        \foreach \i in {0,16,32,48}{
    \foreach \j in {0,16,32,48,64}{
    \draw[fleche] (\i,\j) to ++(16,0);
    }}
            \foreach \i in {0,16,32,48}{
    \foreach \j in {0,16,32,48}{
    \draw[fleche] (\i,\j) to ++(16,16);
    }}
    \foreach \i in {0,16,32,48,64}{
    \foreach \j in {0,16,32,48,64}{
    \draw[fleche] (\i,\j) to ++(0,8);;
    \draw[fleche] (\i,\j) to ++(0,-8);;
        \draw[fleche] (\i,\j) to ++(8,0);;
    \draw[fleche] (\i,\j) to ++(-8,0);;
        \draw[fleche] (\i,\j) to ++(8,8);;
    \draw[fleche] (\i,\j) to ++(-8,-8);;
    }}
	\end{tikzpicture}
		\end{center}
		\caption{triangular lattice}
	\end{subfigure}
	\caption{Different graph structures can be obtained from the vertex set $\Z^2$ by varying $\omega$ and thus the edge set. The regular lattice is obtained by choosing $\omega$ as in \eqref{eq:standard-diffusion-vector} (see left figure) and the triangular lattice is obtained by choosing $\omega$ as in \eqref{eq:triangular} (see right figure).}
	\label{fig:drawing-graphs}
\end{figure}

The synchronous simple gossip algorithm iterates the local averaging operation:
\begin{align}
\label{eq:grid-synchronous-gossip}
x_{n+1}(v) = (\omega * x_n)(v) = \sum_{\eta \in \Z^d} \omega(v-\eta) x_n(\eta) \, , \qquad v \in \Z^d \, .
\end{align}
Note that this equation is \eqref{eq:simple-gossip-general} in the case where the averaging weights $W_{v-\eta} = \omega(v-\eta)$ are translation-invariant. 

\medskip
Throughout this paper, we assume that $\omega$ is centered, i.e., 
\begin{align*}
\sum_{v\in \Z^d} \omega(v) v = 0 \, ,
\end{align*}
and that the covariance matrix 
\begin{align}
\label{eq:cov_matrix}
Q = \sum_{v\in\Z^d} \omega(v) v v^\top \, , 
\end{align}
has full rank. Heuristically, this ensures that we do not have a drift term, and that we average in all directions. 
We note that our approach can still obtain similar results without these assumptions, however, 


\medskip\noindent 
\textbf{Goal.} In the gossip problem, we give a real value $x_0(v)$ to each one of the nodes, and the goal is to compute the average of those values. In this paper, we consider graphs with an infinite number of nodes; the average of the infinite number of initial values is ill-defined. Thus the gossip problem is misspecified in this context. However, the gossip iterations like \eqref{eq:simple-gossip-general} make sense even on an infinite graph. In the following, we study the scaling limits of these iterations. This provides an intuition on the behavior on large but finite graphs. 

More precisely, on $\Z^d$, we study the decay to $0$ of the gossip algorithms when initialized from $x_0 = \bfone_0$, the vector with entry $\bfone_0(0)=1$ and all other entries $\bfone_0(v)$, $v \in \Z^d\backslash\{0\}$, equal to~$0$. By analogy with PDEs, the solution we obtain with this initial condition is the \emph{fundamental solution} of the gossip iterations; the solutions for other initializations $x_0$ can be obtained by convolution of $x_0$ with the fundamental solution.

\section{Heuristic Derivations and Discussions}
\label{sec:heuristic}

Before we go into the rigorous convergence results, we would like to provide a heuristic derivation. 
While these are not proofs --- they fail whenever the solution is not sufficiently smooth for the required Taylor expansions --- we find these provide very strong intuition as to why the gossip algorithms converge to the corresponding PDEs. 

\subsection{Scaling limit of the simple gossip algorithm to the heat equation}
\label{sec:scaling}

Through a change of variables in \eqref{eq:grid-synchronous-gossip}, we rewrite the synchronous simple gossip algorithm iterates as
\begin{align}
\label{eq:aux-1}
x_{n+1}(v) =  \sum_{\eta \in \Z^d} \omega(\eta) x_n(v-\eta) \, , \qquad v \in \Z^d \, .
\end{align}
Let $\Delta t, \Delta y > 0$ denote two scaling parameters. For $t \in (\Delta t) \N = \{(\Delta t) n, \, n \in \N\}$ and $y \in (\Delta y) \Z^d = \{(\Delta y) v , \, v \in \Z^d\}$, we define the scaled field 
\begin{align*}
u(t,y) = x_{\frac{t}{\Delta t}}\left(\frac{y}{\Delta y}\right) \, . 
\end{align*}
The iteration~\eqref{eq:aux-1} can be reformulated in terms of $u$:
\begin{equation*}
u(t + \Delta t, y) = \sum_{\eta \in \Z^d} \omega(\eta) u(t, y - (\Delta y) \eta) \, .
\end{equation*}
We now show that under a proper scaling for $\Delta t, \Delta y \to 0$, the above equation converges to a PDE in $u$. Recall that $\sum_{v \in \Z^d} \omega(v) = 1$, thus
\begin{equation*}
u(t + \Delta t, y) - u(t,y) = \sum_{\eta \in \Z^d} \omega(\eta) \left[ u(t, y - (\Delta y) \eta) - u(t,y) \right]\, .
\end{equation*}

Before we take $\Delta t, \Delta y \to 0$, we first Taylor expand the differences: 
\begin{align*}
u(t + \Delta t, y) - u(t,y) &= (\Delta t) \partial_t u + o(\Delta t) \, , \\
u(t, y - (\Delta y) \eta) - u(t,y) &= - (\Delta y) \left\langle \nabla_y u, \eta \right\rangle + \frac{(\Delta y)^2}{2} \left\langle \eta, \left(\nabla^2_y u\right)  \eta \right\rangle + o\left((\Delta y)^2\right) \, ,
\end{align*}
where all derivatives are taken in $(t,y)$. Note that we make a second-order expansion in space: this is due to the fact that the first-order terms cancel below. We obtain 
\begin{align*}
&(\Delta t) \partial_t u + o(\Delta t) \\
&\qquad = - (\Delta y) \left\langle \nabla_y u, \sum_{\eta \in \Z^d} \omega(\eta) \eta \right\rangle + \frac{(\Delta y)^2}{2} \sum_{\eta \in \Z^d} \omega(\eta)\left\langle \eta, \left(\nabla^2_y u\right)  \eta \right\rangle + o\left((\Delta y)^2\right) \, .
\end{align*}
As $\omega$ is centered, the first term of the right-hand side is zero. Moreover, we can use \eqref{eq:cov_matrix} to rewrite 
\begin{align*}
\sum_{\eta \in \Z^d} \omega(\eta)\left\langle \eta, \left(\nabla^2_y u\right)  \eta \right\rangle = \Tr\left(Q \nabla^2_y u\right) = \nabla_y \cdot (Q \nabla_y u) \, .
\end{align*}
This gives us 
\begin{align*}
(\Delta t) \partial_t u + o(\Delta t) =  \frac{(\Delta y)^2}{2} \nabla_y \cdot (Q \nabla_y u) + o\left((\Delta y)^2\right) \, .
\end{align*}
Finally, we choose the scaling $\Delta t = (\Delta y)^2$ and by identifying the highest-order terms, we obtain the scaling to the heat equation in the limit as $\Delta t, \Delta y \to 0$ 
\begin{align*}
\partial_t u  =   \frac{1}{2}\nabla_y \cdot (Q \nabla_y u)  \, .
\end{align*}
Here, $Q$ quantifies the potential anisotropy of the diffusion. In the case of the standard grid~\eqref{eq:standard-diffusion-vector}, we have $Q = \frac{1}{d} \Id$ and thus we obtain an isotropic heat equation $\partial_t u  =   \frac{1}{2d}\Delta_y u$, where $\Delta_y$ denotes the Laplacian in the variable $y$. 

\subsection{Second-order iteration scaling to the Euler--Poisson--Darboux equation}
\label{sec:second-order}

We now consider second-order iterations of the form 
\begin{align}
\label{eq:second-order-5}
x_{n+1}(v) = a_n \sum_{\eta \in \Z^d} \omega(\eta) x_n(v-\eta) + b_n x_n(v) - c_n x_{n-1}(v) \, .
\end{align}
We impose $a_n + b_n - c_n = 1$ so that the sum of the coordinates of the vectors $x_n$ remains constant. We show that, under specific asymptotics for $a_n,b_n,c_n$, the iteration~\eqref{eq:second-order-5} scales to the EPD equation. As in Section~\ref{sec:scaling}, we introduce scaling parameters $\Delta t, \Delta y > 0$ and the rescaled iterates 
\begin{align*}
u(t,y) = x_{\frac{t}{\Delta t}}\left(\frac{y}{\Delta y}\right) \, . 
\end{align*} 
The iteration~\eqref{eq:second-order-5} can be reformulated in terms of $u$:
\begin{equation*}
u(t + \Delta t, y) = a_n \sum_{\eta \in \Z^d} \omega(\eta) u(t, y - (\Delta y) \eta) +b_n u(t,y) - c_n u(t-\Delta t, y) \, .
\end{equation*}
Subtracting $u(t,y)$ and using $a_n + b_n - c_n = 1$, we obtain 
\begin{align*}
&u(t + \Delta t, y) - u(t,y)\\
&\qquad= a_n \sum_{\eta \in \Z^d} \omega(\eta) \left[ u(t, y - (\Delta y) \eta)- u(t,y) \right]  - c_n \left[u(t-\Delta t, y) - u(t,y) \right] \, .
\end{align*}
We make the Taylor expansions of $u$, but this time a second-order expansion in $t$ is necessary: 
\begin{align*}
u(t + \Delta t, y) - u(t,y) &= (\Delta t) \partial_t u + \frac{(\Delta t)^2}{2} \partial_{tt} u +  o(\Delta t) \, , \\
u(t - \Delta t, y) - u(t,y) &= - (\Delta t) \partial_t u + \frac{(\Delta t)^2}{2} \partial_{tt} u +  o(\Delta t) \, , \\
u(t, y - (\Delta y) \eta) - u(t,y) &= - (\Delta y) \left\langle \nabla_y u, \eta \right\rangle + \frac{(\Delta y)^2}{2} \left\langle \eta, \left(\nabla^2_y u\right)  \eta \right\rangle + o\left((\Delta y)^2\right) \, .
\end{align*}
We obtain 
\begin{align*}
\frac{(\Delta t)^2}{2} (1+c_n) \partial_{tt} u + (\Delta t) (1-c_n) \partial_t u = a_n \frac{(\Delta y)^2}{2} \nabla_y \cdot \left( Q \nabla_y u\right) \, . 
\end{align*}
To have the scaling to the Euler--Poisson--Darboux (EPD) equation, we take $\Delta t = \Delta y$, and
\begin{align}
\label{eq:scaling-coeffs}
&a_n \xrightarrow[n\to\infty]{} 2 \, , &&c_n = 1 - \frac{d+1}{n} + o\left(\frac{1}{n}\right) \, .
\end{align}
Indeed, as $t = n \Delta t$, we have $1 - c_n \sim \frac{d+1}{t}\Delta t$ and thus 
\begin{align*}
\frac{(\Delta t)^2}{2} (2 +o(1)) \partial_{tt} u + (\Delta t)^2 \left(\frac{d+1}{t} + o(1)\right)\partial_t u &= (2 + o(1)) \frac{(\Delta y)^2}{2} \nabla_y \cdot \left( Q \nabla_y u\right) \, , 
\end{align*}
thus by choosing the scaling $\Delta t = \Delta y$ and identifying highest-order terms, we obtain the EPD equation in the limit as $\Delta t, \Delta y \to 0$: 
\begin{align*}
\partial_{tt} u + \frac{d+1}{t} \partial_t u &= \nabla_y \cdot \left(Q\nabla_y u\right) \, .
\end{align*}
Note that there is the implicit condition $b_n \xrightarrow[]{} 0$ implied by~\eqref{eq:scaling-coeffs} as $a_n + b_n - c_n = 1$. 

\medskip\noindent
\textbf{Different scalings.} Note that in this section, the scaling is $\Delta t = \Delta y$ while for the simple gossip, the scaling is $\Delta t = (\Delta y)^2$. This is another illustration that the iteration of this section diffuses faster: to scale to a non-degenerate object, it needs go though a higher order rescaling in space. 

\subsection{Probabilistic interpretation}

For the sake of mathematical curiosity, let us make an aside on the probabilistic interpretations of the heat equation and of the EPD equation. It is well-known that the heat equation represents the evolution of the probability density function of Brownian motion in $\R^d$ \citep{legall2018}. 
Stochastic representation for wave equations remains largely an open problem \citep{dalang2008feynman,chatterjee2013stochastic}; however \citet{kac1974stochastic} has shown that in dimension $d=1$, the solutions of \eqref{eq:epd}
represent the evolution of a persistent random walk: $u(t,.)$ is the probability density function of a random walker in $\R$, that moves according to a speed $+1$ or $-1$, and, at a time-dependent Poisson rate $a(t) = 1/t$, reverses its speed. Thus, the EPD equation~\eqref{eq:epd} is the density of a persistent random walk that gets more and more persistent over time. The rate $a(t) = 1/t$ of the speed resampling is chosen so that the law of the random walker, when started from $0$, is uniform on the interval $[-t,t]$. Note that as the random walker has unit speed, it can not be at a distance larger than the elapsed time $t$ from the starting point. 

This probabilistic point of view gives further credence to the high-level idea that acceleration is achieved by giving inertia to the gossiped information. 

\subsection{Relation to the Jacobi polynomial iteration}
\label{sec:relation-jacobi}

For the convenience of the reader, we recall here the Jacobi polynomial iteration introduced by \citet{berthier2020accelerated} to accelerate gossip algorithms:
\begin{align}
&x_1 = a_0 \omega * x_0 + b_0 x_0 \, , &&x_{n+1} = a_n \omega *  x_n + b_n x_n - c_n x_{n-1} , \label{eq:jpi-1}\\
&a_0 = \frac{d+4}{2(2+d)} \, , &&b_0 = \frac{d}{2(2+d)} \, , \\
&a_n = \frac{(2n+d/2+1)(2n+d/2+2)}{2(n+1+d/2)^2} \, , &&b_n = \frac{d^2(2n+d/2+1)}{8(n+1+d/2)^2(2n+d/2)} \, , \\
&c_n = \frac{n^2(2n+d/2+2)}{(n+1+d/2)^2(2n+d/2)} \, , \qquad n \geq 1 \, . \label{eq:jpi-4}
\end{align} 
As explained in \citep{berthier2020accelerated}, this iteration is associated to the Jacobi polynomials $P_n^{(\alpha,\beta)}$ with parameters $\alpha = d/2$ and $\beta = 0$. It is of the form~\eqref{eq:second-order-5} with coefficients satisfying~\eqref{eq:scaling-coeffs}. Thus the Jacobi polynomial iteration scales to the EPD equation. However, note the large difference between the approaches of \citep{berthier2020accelerated} and this paper: in \citep{berthier2020accelerated}, the authors use the geometry of the graph to approximate the spectrum of the gossip problem and design a polynomial-based method adapted to this approximate spectrum; in this paper, we also use the geometry of the graph but to view gossip algorithms as PDEs when rescaled. It is remarkable that the two approaches lead to similar results. 

The PDE perspective enriches our understanding of the Jacobi polynomial iteration. For instance, one can explore the effect of using the Jacobi polynomial $P_n^{(\alpha,\beta)}$ for a different value than $(\alpha,\beta) = (d/2,0)$ used in the Jacobi polynomial iteration. The formula~\citep[Equation (SM6.2)]{berthier2020accelerated} gives the expression of the recurrence coefficients of the Jacobi polynomial iteration in this general case; it follows that
\begin{align*}
&a_n \xrightarrow[n\to\infty]{} 2 \, , &&c_n = 1 - \frac{2\alpha+1}{n} + o\left(\frac{1}{n}\right) \, ,
\end{align*}
thus, repeating the computations of Section~\ref{sec:second-order}, the iteration converges to the more general EPD equation~\eqref{eq:epd-general}. Consider its fundamental solution~\eqref{eq:fundamental-solution}. If $\alpha>d/2$, the mass concentrates at the center of the ball of radius $t$. On the contrary, if $\alpha < d/2$, the mass concentrates at the edge of the ball. Both effects are undesirable as uniform averaging is the optimal strategy. These effects are simulated in Figure~\ref{fig:scaling-alpha}.

\begin{figure}
	\begin{subfigure}{0.7\linewidth}
		\begin{center}
			\includegraphics[width = \linewidth]{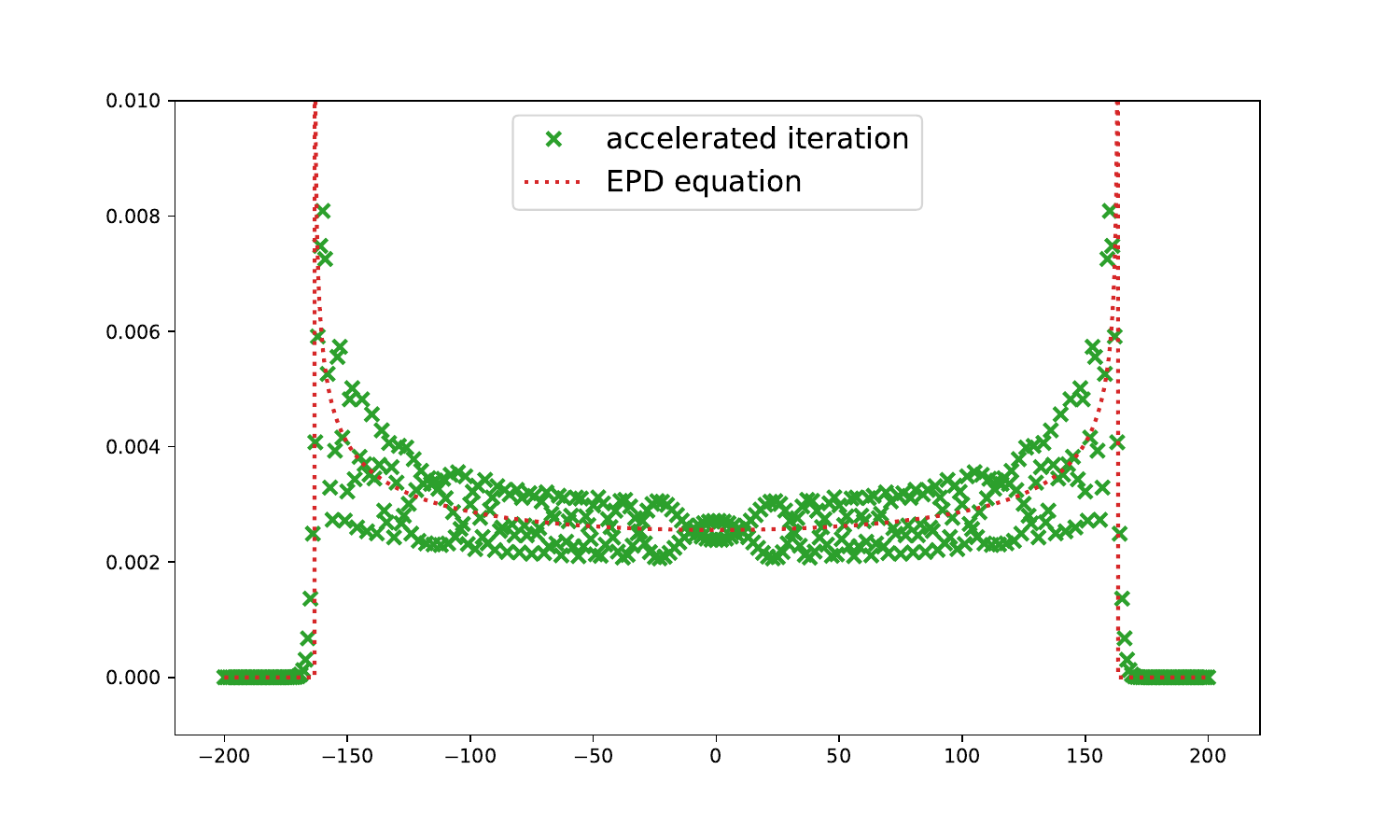}
		\end{center}
		\caption{$\alpha = \frac{1}{4} < \frac{1}{2} = \frac{d}{2}$}
	\end{subfigure}
	\begin{subfigure}{0.7\linewidth}
		\begin{center}
			\includegraphics[width = \linewidth]{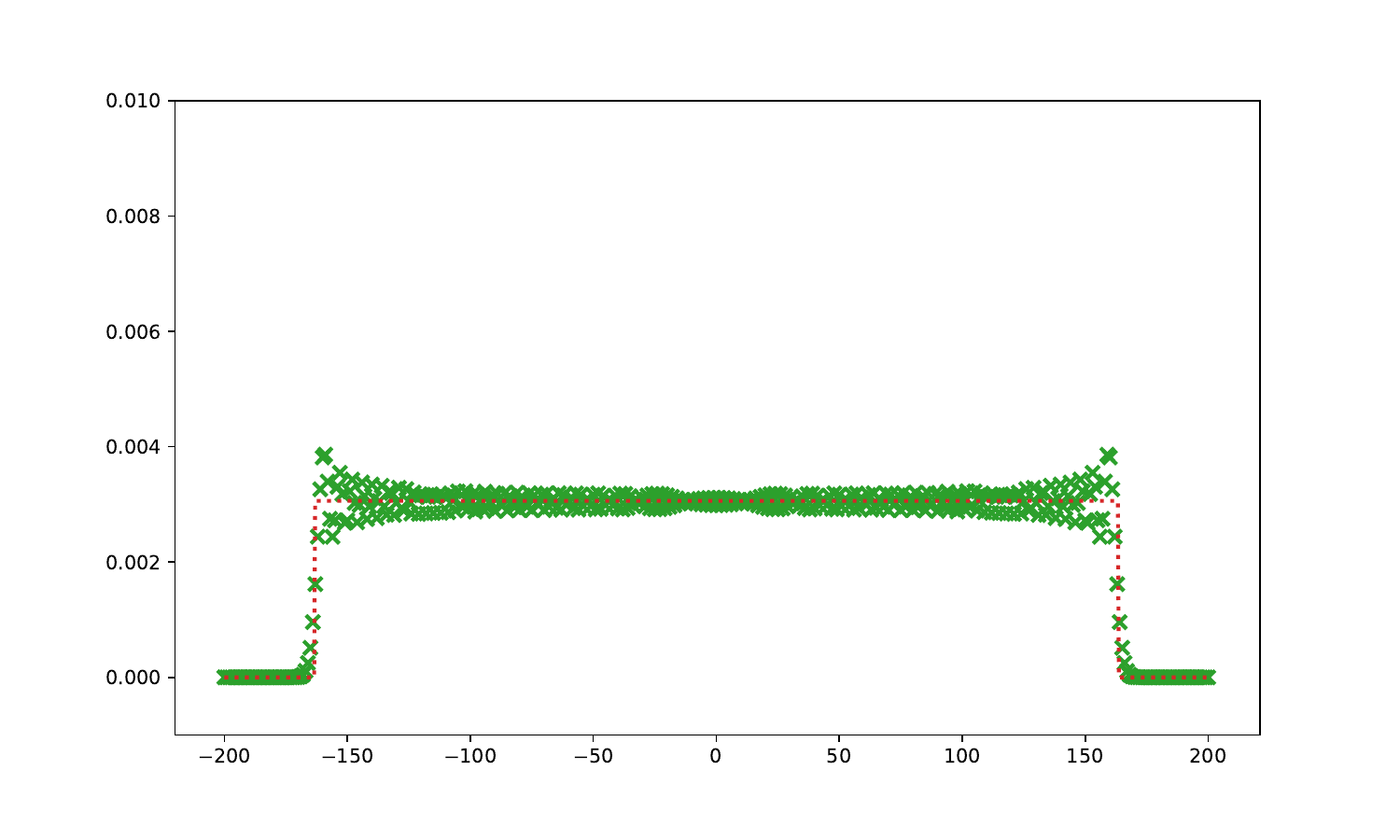}
		\end{center}
		\caption{$\alpha = \frac{1}{2} = \frac{d}{2}$}
	\end{subfigure}
	\begin{subfigure}{0.7\linewidth}
		\begin{center}
			\includegraphics[width = \linewidth]{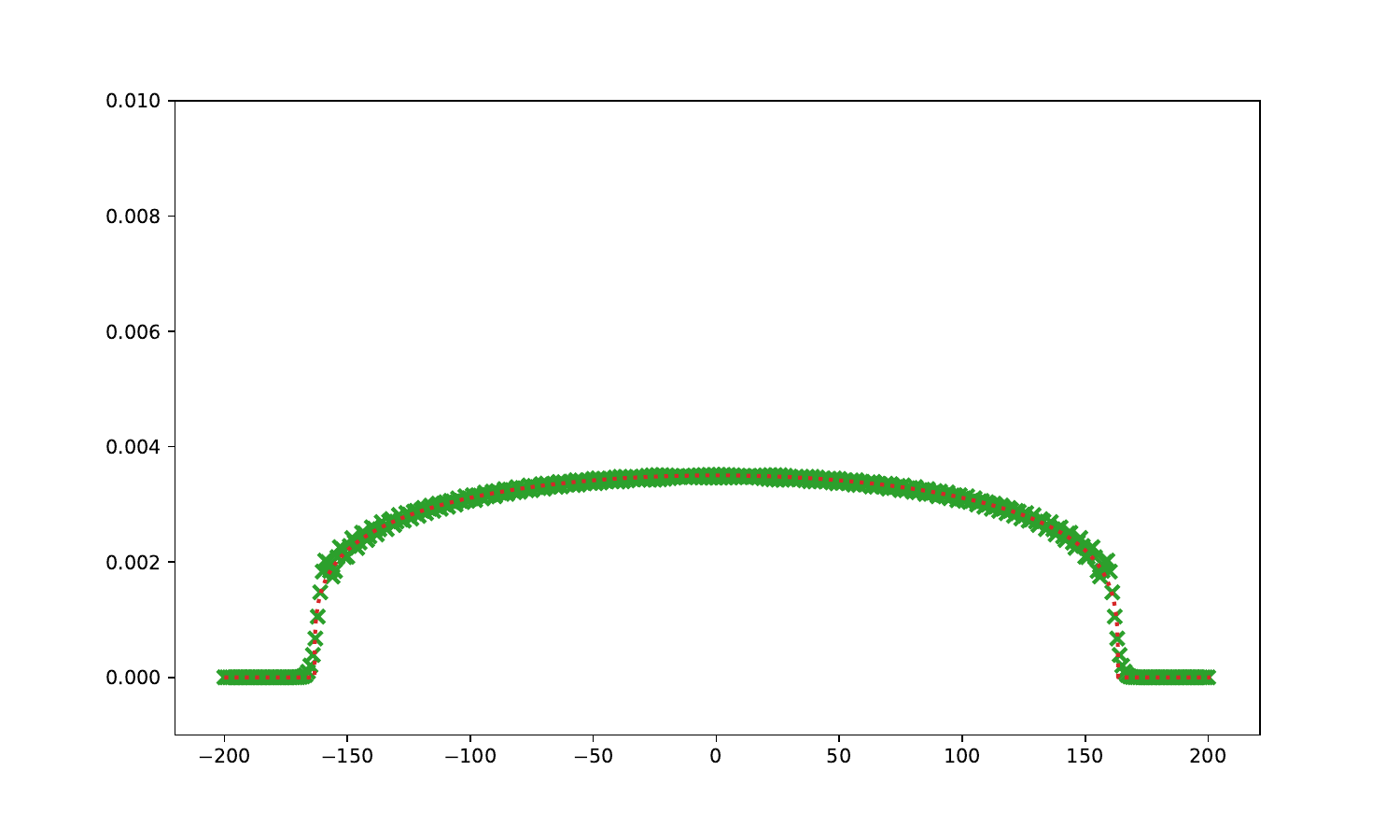}
		\end{center}
		\caption{$\alpha = \frac{3}{4} > \frac{1}{2} =  \frac{d}{2}$}
	\end{subfigure}
	\caption{Same simulation as in Figure~\ref{fig:scaling}(C), but we now study the effect of varying the parameter $\alpha$ of the Jacobi polynomial iteration. Varying $\alpha$ also changes the fundamental solution~\eqref{eq:fundamental-solution} of the EPD equation~\eqref{eq:epd-general}.}
	\label{fig:scaling-alpha}
\end{figure}

\subsection{Related numerical PDE methods and problems}

Both the simple gossip algorithm and Jacobi polynomial iteration can be interpreted as numerical discretizations of their limiting PDEs, which has a rich literature studying both theory and implementation \citep{celia1992numerical}. 
To our best knowledge, the only theoretical treatment of numerical discretization errors for the EPD equation is due to \cite{genis1984finite}, where the author studied a large class of finite element methods (FEMs). 
Here we note the local averaging operation used in gossip algorithms can be seen as a special case of FEM, where averaging is approximating the Laplacian operator. 

However, we emphasize two key differences. 
Firstly, FEM generally does not preserve total mass, which is a key property of gossip algorithms. 
In this sense, the Jacobi polynomial iteration is closer to a finite volume method (FVM) \citep{leveque2002finite}, which typically preserves a quantity of interest such as volume. 
Secondly, most numerical analysis approaches require some level of regularity for the initial condition, which cannot apply to the fundamental solution (with Dirac delta initial condition). 
In the case of \cite{genis1984finite}, all of the error bounds were in terms of an equivalent fractional Sobolev norm of the initial condition, which is unbounded for the Dirac delta. 
Instead, we were able to avoid regularity issues by studying the discretization error in the Fourier domain, see Section \ref{sec:rigorous} and Appendix \ref{ap:proofs-epd}. 

There is a body of empirical work on numerical methods for EPD and related equations (for example \cite{glowinski2013euler}), however these generally cannot be transformed to a locally communicating algorithm, which is important to build practical gossip algorithms. 
We also mention a Hamilton--Jacobi equation of the type $\partial_t u = \| \nabla u \|$, which has a similar fundamental solution of the form  $\bfone_{\left\{ \|y\| \leq t \right\}}$. 
It is well known this equation can be solved by the fast sweeping method \citep{zhao2005fast} and the fast marching method \citep{sethian1996fast}. 
However, also due to the lack of normalization (to preserve mass), and the fact these algorithms are not locally communicating, they are not good candidates for gossip algorithms.

\subsection{Open problems: other geometries, stochastic case}

An important limitation of this work is that we only study synchronous gossip on a regular lattice. It is natural to ask what could happen in an asynchronous setting, or when the graph is microscopically perturbed (percolation graph, random geometric graph, etc). 

For the simple gossip, or equivalently, for the random walk or for heat diffusion, answering this question is the subject of the field of homogenization, see, e.g., \citep{armstrong2019quantitative,armstrong2018elliptic,biskup2011recent}. The heuristic is that on a large scale and for long diffusion times, microscopic fluctuations of the connectivity (in space and in time) are homogenized: the process scales to a homogeneous diffusion with some constant effective diffusion matrix $Q$.

Our work raises the following question: is there homogenization for the EPD equation? \citet{berthier2020accelerated} prove that there is some robustness of the Jacobi polynomial iteration to microscopic details of the graphs, as the rates are the same on all graphs of spectral dimension $d$. However, we do not know if the process scales to the same limit on those graphs.

\section{Rigorous convergence results}
\label{sec:rigorous}

In this section, we provide rigorous justification to the heuristic derivations of Section~\ref{sec:heuristic}. In Section~\ref{sec:rigorous-simple}, we start with the convergence of simple gossip to the heat equation. This case is simple as it is equivalent to the central limit theorem: we obtain a weak convergence result. A stronger convergence result, of local type, is deduced from the local central limit theorem. 

Section~\ref{sec:rigorous-simple} illustrates that two types of convergence are possible: weak and local. In Section~\ref{sec:rigorous-jacobi}, we prove analog results for the convergence of the Jacobi polynomial iteration to the EPD equation. We restrict ourselves to the Jacobi polynomial iteration---and not to any method satisfying~\eqref{eq:scaling-coeffs}---for technical reasons: we use fine asymptotic properties of the Jacobi polynomials. However, we end this section with a remark on why we conjecture the same scaling for all iterations satisfying~\eqref{eq:scaling-coeffs}.

In Section~\ref{sec:sharp-rates}, we apply the local convergence result to obtain convergence rates of the Jacobi polynomial iteration. These rates are sharp up to constants.

\subsection{Simple gossip and the heat equation}
\label{sec:rigorous-simple}

Consider the simple gossip iteration 
\begin{align*}
&x_0 = \bfone_0 \, , &&x_{n+1} = \omega * x_n \, .
\end{align*}
The iteration $x_n$ can be interpreted as the probability density function of a random walk on $\Z^d$, initialized from $0$, with increments of law $\omega$. As $\omega$ is centered, the random walk is unbiased; the matrix $Q$ is the covariance of the increments. The asymptotic law $x_n$ is described by the central limit theorems: here, we interpret them with our notations. Let $u(t,y)$ denote the fundamental solution~\eqref{eq:solution-heat} of the heat equation~\eqref{eq:heat-equation}. We denote $\delta_y$ the Dirac mass at $y \in \R^d$.
\begin{thm}[Central Limit Theorem, see, e.g., \citep{billingsley2008probability}]
	\label{thm:clt}
	We have the following weak convergence in the space of positive measures: for any $t\geq0$,
	\begin{align*}
	\sum_{v \in \Z^d} x_{\lfloor t/\varepsilon^2\rfloor}(v) \delta_{\varepsilon v} \xrightarrow[\varepsilon \to 0]{} u(t,y)\diff y \, .
	\end{align*}
\end{thm}
A stronger local result holds assuming that $\omega$ is aperiodic, i.e., that the random walk with increments $\omega$ is an aperiodic Markov chain on $\Z^d$ \citep[Section~8]{billingsley2008probability}. For instance, the vector $\omega$ of Equation \eqref{eq:triangular}, corresponding to the triangular lattice, is aperiodic, while the vector $\omega$ of Equation \eqref{eq:standard-diffusion-vector}, corresponding to the regular grid, is not.

\begin{thm}[Local Central Limit Theorem, \citep{gnedenko1948local}]
	\label{thm:local-clt}
	Assume that $\omega$ is aperiodic. Then
	\begin{align*}
	\sup_{v\in\Z^d} \left\vert x_n(v) - u(n,v)\right\vert = o\left(\frac{1}{n^{d/2}}\right) \qquad \text{as }n \to \infty. 
	\end{align*}
\end{thm}
A pedagogical introduction to the local central limit theorem is provided by \citet{curien2020random}. The beauty of the local central limit theorem is that no rescaling is required: we simply discretize the heat equation in time and space. 

\subsection{The Jacobi polynomial iteration and the Euler--Poisson--Darboux equation}
\label{sec:rigorous-jacobi}

We now give analogs of Theorems~\ref{thm:clt} and~\ref{thm:local-clt} for the convergence of the Jacobi polynomial iteration to the EPD equation. Let $x_n$ denote the iterates of the Jacobi polynomial iteration~\eqref{eq:jpi-1}-\eqref{eq:jpi-4} initialized from $x_0 = \bfone_0$ and $u(t,y)$ the fundamental solution~\eqref{eq:fundamental-solution-d/2} of the EPD equation~\eqref{eq:epd}. 

\medskip\noindent
\textbf{Assumptions.} In this section, we assume that $\omega$ is symmetric ($\omega(-v) = \omega(v)$) and aperiodic. While the aperiodicity assumption is clearly necessary for Theorem~\ref{thm:local-jpi} to hold, we do not know if these assumptions are necessary otherwise.

\begin{thm}[Weak Convergence]
	\label{thm:weak-epd}
	We have the following weak convergence in the space of signed measures: for all $t > 0$,
	\begin{equation*}
	\sum_{v\in\Z^d} x_{\lfloor t/\varepsilon \rfloor}(v) \delta_{\varepsilon v} \xrightarrow[\varepsilon \to 0]{} u(t,y) \diff y \, .
	\end{equation*}
\end{thm}

\begin{thm}[Local Convergence]
\label{thm:local-jpi}
Let $\psi(x) := \prod_{i=1}^d \frac{\sin(\pi x_i)}{ \pi x_i }$. Then we have that 
\begin{align*}
\sum_{v \in \Z^d} 
\left( x_n(v) - (u(n, \cdot ) * \psi)(v) \right)^2 = o\left(\frac{1}{n^d}\right) \,, 
\qquad \text{as }n \to \infty. 
\end{align*}
\end{thm}


The two theorems are proved in Appendix~\ref{ap:proofs-epd}. 

\begin{remark}
    We believe that Theorem \ref{thm:local-jpi} should hold without the convolution with $\psi$, namely, 
	$\sum_{v \in \Z^d} \left( x_n(v) - \left(u(n,v) \right)(v)\right)^2 = o\left(\frac{1}{n^d}\right)$, 
	especially in light of Theorem \ref{thm:weak-epd}. 
	However, there are some technical challenges with the proof that we are unable to resolve.
\end{remark}

\begin{remark}
	The statements of Theorems~\ref{thm:weak-epd} and~\ref{thm:local-jpi} and their proofs can be easily adapted to study the Jacobi polynomial iterations for other parameters $(\alpha,\beta) \neq (d/2,0)$, as long as $\alpha > d/2-1/2$ and $\beta \leq \alpha$. In this case, the limiting PDE depends on $\alpha$. We have the convergence to the fundamental solution~\eqref{eq:fundamental-solution} of the general EPD equation~\eqref{eq:epd-general}. 
\end{remark}

\begin{remark}[Extension beyond the Jacobi polynomial iteration]
	Our theorems are stated for the Jacobi polynomial iteration only for a technical reason: the proofs are based on well-known asymptotic properties of the Jacobi polynomials, stated in Proposition~\ref{prop:jacobi}. We conjecture that all other sequences of polynomials with recursion coefficients satisfying~\eqref{eq:scaling-coeffs} also satisfy the same properties: this would prove the scaling to the EPD equation for all second-order gossip algorithms satisfying~\eqref{eq:scaling-coeffs}. 
	
	This conjecture is supported by \citet{aptekarev1993asymptotics}: he shows that a sequence of orthogonal polynomial must satisfy the Mehler--Heine asymptotics (Proposition~\ref{prop:jacobi}.(\ref{it:mehler-heine})) provided that the recurrence coefficients of the polynomials satisfy some conditions that resemble~\eqref{eq:scaling-coeffs}. Interestingly, he explains that the asymptotics of the recurrence coefficients are related to the shape of the orthogonality measure of the associated orthogonal polynomials near $1$: this links the approaches of \citep{berthier2020accelerated} and this paper. 
\end{remark}

\subsection{Application: sharp rates of the Jacobi polynomial iteration on $\Z^d$}
\label{sec:sharp-rates}

In this section, we apply Theorem~\ref{thm:local-jpi} to obtain sharp rates for the Jacobi polynomial iteration. 

\begin{coro}
	\label{coro:sharp-rates}
	Assume that $\omega$ is symmetric and aperiodic. Let $x_n$ be the iterates of the Jacobi polynomial iteration~\eqref{eq:jpi-1}-\eqref{eq:jpi-4}, initialized at $x_0 = \bfone_0$. Then we have the asymptotic equivalence
	\begin{align*}
	\sum_{v \in \Z^d} x_n(v)^2 \underset{n\to\infty}{\sim} \frac{1}{(\det Q)^{1/2} |B(0,1)|} \frac{1}{n^d} \, ,
	\end{align*}
	where $|B(0,1)| = \frac{\pi^{d/2}}{\Gamma\left(d/2+1\right)}$ is the volume of the Euclidean unit ball in dimension $d$.
\end{coro}

We compare with \citet[Theorem SM7.1]{berthier2020accelerated}. Here our theorem applies only to regular lattices, while the previous result applies to all graphs of spectral dimension $d$; but we obtain an asymptotic equivalent, while the previous result gave only the exponent in $n$. In Figure~\ref{fig:sharp-rates}, we compare the two asymptotic equivalent quantities in the case of the Jacobi polynomial iteration on the triangular lattice. Note that similarly, one could obtain sharp rates for simple gossip from the local central limit Theorem~\ref{thm:local-clt}. 

\begin{figure}
	\includegraphics[width = 0.6\linewidth]{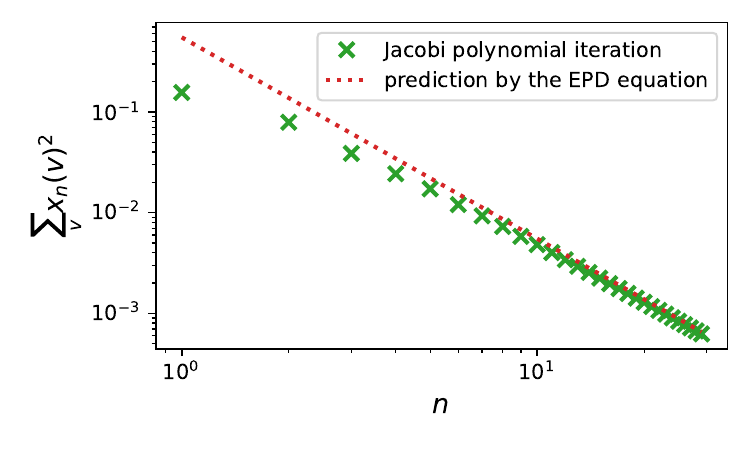}
	\caption{Comparison between the empirical error $\sum_{v\in\Z^d} x_n(v)^2$ and the asymptotic rate predicted by Corollary~\ref{coro:sharp-rates}. Here, $x_n$ are the iterates of the Jacobi polynomial iteration on the triangular lattice~\eqref{eq:triangular}. Note that the corollary predicts sharply not only the scaling in $n$ (the asymptotic slope in the logarithmic plot) but also the leading constant (the intercept of the asymptotic line). }
	\label{fig:sharp-rates}
\end{figure}

\section*{Acknowledgements}

We would like to thank the Saint-Flour Probability Summer School for the opportunity to meet and initiate this project. 
We would also like to thank Patric Bonnier, Christina Christara, Dejan Slep\v{c}ev for insightful discussions. 
ML is supported by the Ontario Graduate Scholarship and the Vector Institute. 

\bibliographystyle{abbrvnat}
\bibliography{bibliography}

\appendix

\newpage

	\section{The Euler--Poisson--Darboux (EPD) equation}
\label{ap:epd}

The EPD equation is the partial differential equation 
\begin{align}
\label{eq:epd-general}
\partial_{tt} u + \frac{2\alpha+1}{t} \partial_t u = \nabla_y \cdot \left(Q\nabla_y u\right) \, .
\end{align}
Posing a rigorous framework for solving this equation is subtle because there is a diverging coefficient $\frac{2\alpha+1}{t}$ as $t \to 0$. Moreover, we see below that fundamental solutions are irregular; they are defined in a weak sense. Thankfully, we do not have to bother with these technical details as our rigorous results only require to know the expression of the fundamental solution of the EPD equation (Proposition~\ref{prop:fundamental-solution}) and its Fourier transform (Proposition~\ref{prop:fourier-fundamental-solution}). These expressions are given by \citet{bresters1973equation} in the case $Q = \Id$; here, we easily extend the expressions for a general matrix $Q$.

\begin{proposition}
	\label{prop:fundamental-solution}
	The fundamental solution of the EPD equation, i.e., the solution initialized from $u(0,.) = \delta_0$, $\partial_t u(0,.) = 0$, is 
	\begin{align}
	\label{eq:fundamental-solution}
	u(t,y) = \frac{\Gamma(\alpha+1)}{\pi^{d/2}\Gamma(\alpha+1-d/2)(\det Q)^{1/2}} \frac{1}{t^{2\alpha}} \left(t^2 - \left\langle y, Q^{-1} y \right\rangle \right)_+^{\alpha-d/2} \, ,
	\end{align}
	where $(.)_+$ denotes the positive part of a real number. 
\end{proposition}
The case $\alpha = d/2$ is particularly important to us; in this case we recover~\eqref{eq:fundamental-solution-d/2} from~\eqref{eq:fundamental-solution}.

\begin{proof}[Proof of Proposition~\ref{prop:fundamental-solution}]
	In the case $Q = \Id$, the solution is given by \citet[Equation (3.4)]{bresters1973equation}: 
	\begin{align*}
	u(t,y) = \frac{\Gamma(\alpha+1)}{\pi^{d/2}\Gamma(\alpha+1-d/2)} \frac{1}{t^{2\alpha}} \left(t^2 -\Vert y \Vert^2 \right)_+^{\alpha-d/2} \, .
	\end{align*}
	In the general case, consider $v(t,y) = u(t,Q^{1/2}y) (\det Q)^{1/2}$. Computations give that $v(t,y)$ is the fundamental solution of the EPD equation~\eqref{eq:epd-general} with $Q = \Id$, thus 
	\begin{equation*}
	u(t,Q^{1/2}y) (\det Q)^{1/2} = v(t,y) = \frac{\Gamma(\alpha+1)}{\pi^{d/2}\Gamma(\alpha+1-d/2)} \frac{1}{t^{2\alpha}} \left(t^2 -\Vert y \Vert^2 \right)_+^{\alpha-d/2} \, .
	\end{equation*}
	This gives the desired formula. 
\end{proof}

\begin{proposition}
	\label{prop:fourier-fundamental-solution}
	The Fourier transform in space of the fundamental solution~\eqref{eq:fundamental-solution} is
	\begin{align*}
	\hat{u}(t,\xi) = \int_{\R^d} \diff y \, e^{i\langle \xi, y \rangle} u(t,y)  = 2^\alpha \Gamma(\alpha+1)  \left\langle \xi, Q\xi\right\rangle^{-\alpha/2} t^{-\alpha} J_\alpha\left(t \left\langle\xi, Q \xi\right\rangle^{1/2}\right) \, ,
	\end{align*}
	where $J_\alpha$ denotes the Bessel function of the first kind of order $\alpha$ \citep[Section 1.71]{szeg1939orthogonal}.
\end{proposition}

\begin{proof}
	In the case $Q = \Id$, the result is given by \citet[Equation (3.1)]{bresters1973equation}:
	\begin{align*}
	\hat{u}(t,\xi)  = 2^\alpha \Gamma(\alpha+1) \Vert \xi \Vert^{-\alpha} t^{-\alpha} J_\alpha\left(t \Vert \xi \Vert\right) \, .
	\end{align*}
	In the general case, $v(t,y) = u(t,Q^{1/2}y) (\det Q)^{1/2}$ is a solution of the EPD equation~\eqref{eq:epd-general} with $Q = \Id$. Moreover, 
	\begin{align*}
	\hat{v}(t,\xi) &=\int_{\R^d} \diff y \, e^{i\langle \xi, y \rangle} v(t,y) \\
	&=  (\det Q)^{1/2} \int_{\R^d} \diff y \, e^{i\langle \xi, y \rangle} u(t,Q^{1/2}y) \, .
	\end{align*}
	In the last integral, we change the variable to $x = Q^{1/2} y$. Then $\diff x = \det(Q^{1/2}) \diff y = (\det Q)^{1/2} \diff y$.
	\begin{align*}
	\hat{v}(t,\xi) &=\int_{\R^d} \diff x \, e^{i\langle \xi, Q^{-1/2} x \rangle} u(t,x) \\ 
	&= \hat{u}(t, Q^{-1/2}\xi) \, .
	\end{align*}
	Thus 
	\begin{align*}
	\hat{u}(t, \xi) &= \hat{v}(t,Q^{1/2}\xi) = 2^\alpha \Gamma(\alpha+1) \left\langle \xi, Q\xi\right\rangle^{-\alpha/2} t^{-\alpha} J_\alpha\left(t \left\langle\xi, Q \xi\right\rangle^{1/2}\right) \, .
	\end{align*}
\end{proof}

\section{Proof of Theorems~\ref{thm:weak-epd} and~\ref{thm:local-jpi}}
\label{ap:proofs-epd}

In the following, we use the notation of \citet{berthier2020accelerated}: $P_n^{(\alpha,\beta)}$ denotes the classical Jacobi polynomials, orthogonal for the Jacobi measure $\diff \sigma(\lambda) = (1 - \lambda)^\alpha (1+\lambda)^\beta$ \citep{szeg1939orthogonal}. We also denote $\pi_n^{(\alpha,\beta)} = P_n^{(\alpha,\beta)}/P_n^{(\alpha,\beta)}(1)$ the Jacobi polynomial rescaled so that $\pi_n^{(\alpha,\beta)}(1) = 1$. 

Consider the algebra $(\ell^2(\Z^d),+,*)$ where $+$ denotes the classical addition over functions and $*$ denotes the convolution. As $\omega \in \ell^2(\Z^d)$, this algebra gives a meaning to the evaluation $P(\omega) \in \ell^2(\Z^d)$ of a polynomial $P$ at $\omega$. By its definition from \citet{berthier2020accelerated}, the Jacobi polynomial iteration \eqref{eq:jpi-1}-\eqref{eq:jpi-4} can be written more compactly as 
\begin{equation*}
    x_n = \pi_n^{(d/2,0)}(\omega) \, . 
\end{equation*}

Fourier analysis plays a central role in the proofs below. If $x \in \ell^2(\Z^d)$, the Fourier transform $\hat{x} \in L^2([-\pi,\pi]^d)$ is defined as $\hat{x}(\xi) = \sum_{v \in \Z^d} e^{i\langle \xi, v \rangle} x(v)$. Consider the algebra $(L^2([-\pi,\pi]^d),+,\cdot)$ where $+$ denotes the classical addition over functions and $\cdot$ denotes the pointwise multiplication of functions. As $\hat{\omega} \in L^2([-\pi,\pi]^d)$, this algebra gives a meaning to the evaluation $P(\hat{\omega}) \in L^2([-\pi,\pi]^d)$ of a polynomial $P$ at $\hat{\omega}$.

The Fourier transform of a sum is the sum of the Fourier transforms, and the Fourier transform of a convolution is the pointwise product of the Fourier transforms, thus if $P$ is a polynomial, $\widehat{P(\omega)} = P(\hat{\omega})$. In particular, in the following, we analyze the Jacobi polynomial iteration by using the relation 
\begin{equation*}
    \hat{x}_n = \widehat{\pi_n^{(d/2,0)}(\omega)} = \pi_n^{(d/2,0)}(\hat{\omega}) \, . 
\end{equation*}

The proofs below use the following well-known results on Jacobi polynomials.

\begin{proposition}
	\label{prop:jacobi}
	\begin{enumerate}
		\item \label{it:mehler-heine}(Mehler--Heine asymptotic) The Jacobi polynomials satisfy the following asymptotic at the edge of the orthogonality measure
		\begin{align*}
		\lim_{n \to \infty} \pi_n^{(d/2,0)}\left(1-\frac{z^2}{2n^2}\right) = 2^{d/2} \Gamma\left(\frac{d}{2}+1\right) z^{-d/2} J_{d/2}(z) \, , 
		\end{align*}
		where $J_{d/2}$ denotes the Bessel function of the first kind of order $d/2$ \citep[Section 1.71]{szeg1939orthogonal}. The convergence is uniform for $z$ in compact sets. 
		\item \label{it:bound-jacobi}On the whole support of the orthogonality measure, we have the following bounds: there exists constants $C_1, C_2 > 0$ such that for all $n \geq 0$,
		\begin{align*}
		\left\vert \pi_n^{(d/2,0)}(\lambda) \right\vert \leq \begin{cases}
		C_1 \left(\arccos | \lambda | \right)^{-d/2-1/2} n^{-d/2-1/2} &\text{if }|\lambda| \leq 1-\frac{1}{n^2} \, , \\
		C_2 &\text{otherwise.}
		\end{cases}
		\end{align*}
	\end{enumerate}
\end{proposition}

\begin{proof}
	\begin{enumerate}
		\item \citet[Theorem 8.1.1]{szeg1939orthogonal} gives the Mehler--Heine asymptotic for the classical Jacobi polynomials $P_n^{(d/2,0)}$:
		\begin{equation*}
		\lim_{n \to \infty} n^{-d/2}P_n^{(d/2,0)}\left(1-\frac{z^2}{2n^2}\right) =  2^{d/2} z^{-d/2} J_{d/2}(z) \, ,
		\end{equation*}
		with uniform convergence for $z$ in compact sets. As 
		\begin{align*}
		&\pi_n^{(d/2,0)} = \frac{P_n^{(d/2,0)}}{P_n^{(d/2,0)}(1)} = \frac{P_n^{(d/2,0)}}{{n+d/2 \choose n}} \, , &&{n+d/2 \choose n} \sim \frac{n^{d/2}}{\Gamma\left(\frac{d}{2} + 1\right)} \, , 
		\end{align*}
		we obtain the desired formula. 
		\item For $\lambda \geq 0$, this is only a reformulation of \citep[Theorem 7.32.2]{szeg1939orthogonal}. For $\lambda < 0$, we use the symmetry of the Jacobi polynomials $P_n^{(d/2,0)}(\lambda) = (-1)^nP_n^{(0,d/2)}(\lambda)$ \citep[Equation (4.1.3)]{szeg1939orthogonal} and use again \citep[Theorem 7.32.2]{szeg1939orthogonal}. 
	\end{enumerate}
\end{proof}

\subsection{Proof of Weak Convergence}

\begin{proof}[Proof of Theorem~\ref{thm:weak-epd}]
	Denote 
	\begin{align*}
	&\mu_{t,\varepsilon} = \sum_{v\in\Z^d} x_{\lfloor t/\varepsilon \rfloor}(v) \delta_{\varepsilon v} \, , &&\mu_t = u(t,y)\diff y
	\end{align*}
	The proof is based on \cite[Theorem 2.1]{baez1993central}, a variant of L\'evy's theorem for signed measures: in order to prove the weak convergence $\mu_{t,\varepsilon} \to \mu_t$ as $\varepsilon \to 0$, it is sufficient to check that the family of measures $\mu_{t,\varepsilon}$, $\varepsilon > 0$ is tight, bounded in total variation, that we have the pointwise convergence of the Fourier transform $\hat{\mu}_{t,\varepsilon}(\xi) = \int_{\R^d} {\mu}_{t,\varepsilon}(\diff y) e^{i\langle \xi, y \rangle}  \to \hat{\mu}_t(\xi) = \int_{\R^d} {\mu}_{t}(\diff y) e^{i\langle \xi, y \rangle} $ almost everywhere. These three conditions are checked below.
	
	\medskip\noindent
	\textbf{Tightness of $\mu_{t,\varepsilon}$, $\varepsilon > 0$.} $\omega$ has a finite support, thus there exists $R > 0$ such that the support of $\omega$ is included in $B(0,R)$. Then for all $n \geq 0$, the support of $w^{*n} = w * \cdots * w$ (with $n$ terms) is included in $B(0,nR)$. The vector $x_n$ is a linear combination of the $w^{*l}$ for $l \leq n$, thus is also included in $B(0,nR)$. Finally, when rescaling by $\varepsilon$, the support of $\mu_{t,\varepsilon} = \sum_{v\in\Z^d} x_{\lfloor t/\varepsilon \rfloor}(v) \delta_{\varepsilon v}$ is included in $B(0,\varepsilon \lfloor t/\varepsilon \rfloor R) \subset B(0,tR)$. The latter set is independent of $\varepsilon$, thus the family of measures $\mu_{t,\varepsilon}$, $\varepsilon > 0$ is tight. 
	
	\medskip\noindent
	\textbf{Boundedness of $\mu_{t,\varepsilon}$, $\varepsilon > 0$.} Note that $\mu_{t,\varepsilon}(\R^d) = 1$, but as $\mu_{t,\varepsilon}$ is a signed measure, we need to show that the total mass $\Vert \mu_{t,\varepsilon} \Vert = \vert \mu_{t,\varepsilon} \vert(\R^d)$ of the total variation $\vert \mu_{t,\varepsilon} \vert$ is bounded independently of $\varepsilon$. By H\"older's inequality, 
	\begin{equation}
	\label{eq:bound-total-variation}
	\Vert \mu_{t,\varepsilon} \Vert = \Vert x_{\lfloor t/\varepsilon \rfloor} \Vert_{l^1(\Z^d)} \leq \vert \Supp x_{\lfloor t/\varepsilon \rfloor} \vert^{1/2} \Vert x_{\lfloor t/\varepsilon \rfloor} \Vert_{l^2(\Z^d)}^{1/2} \, ,
	\end{equation}
	where $ \vert \Supp x_{\lfloor t/\varepsilon \rfloor} \vert$ denotes the cardinal of the support of $x_{\lfloor t/\varepsilon \rfloor}$. As this support is included in $B(0,\lfloor t/\varepsilon \rfloor R)$, its cardinal can be bounded by the number of integer points in $B(0,\lfloor t/\varepsilon \rfloor R)$. This is dominated by $\varepsilon^{-d}$ as $\varepsilon \to 0$. Thus 
	\begin{equation*}
	\vert \Supp x_{\lfloor t/\varepsilon \rfloor} \vert = O(\varepsilon^{-d}) \, .
	\end{equation*}
	We now bound the second term in~\eqref{eq:bound-total-variation}, namely the norm $\Vert x_{\lfloor t/\varepsilon \rfloor} \Vert_{l^2(\Z^d)}$. By Plancherel identity,
	\begin{align*}
	 	\Vert x_n \Vert^2_{\ell^2(\Z^d)} &= \left\Vert \pi_n^{(d/2,0)}(\omega) \right\Vert^2_{\ell^2(\Z^d)} = \frac{1}{(2\pi)^d} \left\Vert \hat{\pi_n^{(d/2,0)}(\omega)} \right\Vert^2_{L^2([-\pi,\pi]^d)} \\ &=  \frac{1}{(2\pi)^d} \left\Vert \pi_n^{(d/2,0)}(\hat{\omega}) \right\Vert^2_{L^2([-\pi,\pi]^d)} = \frac{1}{(2\pi)^d} \int_{[-\pi,\pi]^d} \diff \xi \left\vert \pi_n^{(d/2,0)}(\hat{\omega}(\xi)) \right\vert^2 \, . 
	\end{align*}
	Here, as $\omega$ is symmetric, $\hat{\omega}(\xi)$ is real. We can use the bounds of Proposition~\ref{prop:jacobi}.\eqref{it:bound-jacobi}. We need to estimate $\hat{\omega}(\xi)$. We use the following lemma.
	\begin{lem}
		\label{lem:curien}
		As $\omega$ is aperiodic, there exists $\lambda > 0$ such that 
		\begin{align*}
		\vert \hat{\omega}(\xi) \vert \leq 1-\lambda \Vert \xi \Vert^2 \, , \qquad \xi \in [-\pi,\pi]^d \, .
		\end{align*}
	\end{lem}
	This lemma is simple and proved by \citet[Section 7.1]{curien2020random}. We now return to our estimate of $\left\vert \pi_n^{(d/2,0)}(\hat{\omega}(\xi)) \right\vert^2$. 
	\begin{itemize}
		\item If $\Vert \xi \Vert \geq \frac{1}{\sqrt{\lambda} n}$, we have $\vert \hat{\omega}(\xi) \vert \leq 1-\lambda \Vert \xi \Vert^2 \leq 1 - \frac{1}{n^2}$. Thus by Proposition~\ref{prop:jacobi}.\eqref{it:bound-jacobi},  
		\begin{align*}
		\left\vert \pi_n^{(d/2,0)}(\hat{\omega}(\xi)) \right\vert &\leq C_1 \left(\arccos | \hat{\omega}(\xi) | \right)^{-d/2-1/2} n^{-d/2-1/2} \\
		&\leq C_1 \left(\arccos \left(1 - \lambda \Vert \xi \Vert^2\right) \right)^{-d/2-1/2} n^{-d/2-1/2} \, .
		\end{align*}
		\item If $\Vert \xi \Vert < \frac{1}{\sqrt{\lambda} n}$, we can only say $\left\vert \pi_n^{(d/2,0)}(\hat{\omega}(\xi)) \right\vert \leq C_2$. 
	\end{itemize}
	Thus 
	\begin{align*}
	\Vert x_n \Vert^2_{\ell^2(\Z^d)} &= \frac{1}{(2\pi)^d} \int_{[-\pi,\pi]^d} \diff \xi \left\vert \pi_n^{(d/2,0)}(\hat{\omega}(\xi)) \right\vert^2 \\
	&\leq C_3 n^{-d-1} \int_{\{\Vert \xi \Vert \geq 1/(\sqrt{\lambda}n)\}} \diff \xi \left(\arccos \left(1 - \lambda \Vert \xi \Vert^2\right) \right)^{-d-1} + C_4 \int_{\{\Vert \xi \Vert < 1/(\sqrt{\lambda}n)\}} \diff \xi 
	\end{align*}
	where we use the notation $C_i$ to denote constants independent of $n$. We use a spherical change of variables in the first integral:
	\begin{align}
	\label{eq:aux5-1}
	\Vert x_n \Vert^2_{\ell^2(\Z^d)} 
	&\leq C_5 n^{-d-1} \int_{1/(\sqrt{\lambda}n)}^{\sqrt{d}\pi} \diff r \, r^{d-1}\left(\arccos \left(1 - \lambda r^2\right) \right)^{-d-1} + C_6 n^{-d} \, .
	\end{align}
	As $r \to 0$, $\arccos(1 - \lambda r^2) \sim \sqrt{2\lambda} r$ and therefore 
	\begin{align*}
	r^{d-1}\left(\arccos \left(1 - \lambda r^2\right) \right)^{-d-1} \sim \sqrt{2} \lambda^{-d/2-1/2} r^{-2} \, .
	\end{align*}
	Thus $r^{d-1}\left(\arccos \left(1 - \lambda r^2\right) \right)^{-d-1}$ is not integrable at $0$. We then have, as $n \to \infty$, 
	\begin{align*}
	\int_{1/(\sqrt{\lambda}n)}^{\sqrt{d}\pi} \diff r \, r^{d-1}\left(\arccos \left(1 - \lambda r^2\right) \right)^{-d-1} \sim C_7 \int_{1/(\sqrt{\lambda}n)}^{\sqrt{d}\pi} \diff r \, r^{-2} \sim C_8 n \, .
	\end{align*}
	Putting back in~\eqref{eq:aux5-1}, we obtain $\Vert x_n \Vert^2_{\ell^2(\Z^d)} = O(n^{-d})$. Finally, getting back to ~\eqref{eq:bound-total-variation}, we obtain as $\varepsilon \to 0$
	\begin{align*}
	\Vert \mu_{t,\varepsilon} \Vert \leq \vert \Supp x_{\lfloor t/\varepsilon \rfloor} \vert^{1/2} \Vert x_{\lfloor t/\varepsilon \rfloor} \Vert_{l^2(\Z^d)}^{1/2} = O(\varepsilon^{-d/2})O\left( \left\lfloor  \frac{t}{\varepsilon}\right\rfloor^{-d/2}\right) = O(1) \, .
	\end{align*}
	This shows that the family of measures $\mu_{t,\varepsilon}$, $\varepsilon > 0$ is bounded in total variation.
	
	\medskip\noindent
	\textbf{Pointwise convergence of the Fourier transform.}
	\begin{align*}
	\hat{\mu}_{t,\varepsilon}(\xi) = \int_{\R^d} \diff{\mu}_{t,\varepsilon}(y) e^{i\langle \xi, y \rangle} = \sum_{v\in\Z^d} x_{\lfloor t/\varepsilon\rfloor}(v) e^{i\langle \xi, \varepsilon v \rangle} = \hat{x}_{\lfloor t/\varepsilon\rfloor}(\varepsilon \xi) = \pi_{\lfloor t/\varepsilon\rfloor}^{(d/2,0)}\left(\hat{\omega}(\varepsilon\xi)\right) \, .
	\end{align*}
	As $\varepsilon \to 0$, 
	\begin{equation*}
	\hat{\omega}(\varepsilon\xi) = 1-\frac{\varepsilon^2}{2} \left\langle \xi, Q \xi \right\rangle + o(\varepsilon^2) \, .
	\end{equation*}
	We now apply Proposition~\ref{prop:jacobi}.\eqref{it:mehler-heine}: 
	\begin{align*}
	\hat{\mu}_{t,\varepsilon}(\xi) &= \pi_{\lfloor t/\varepsilon\rfloor}^{(d/2,0)}\left(1-\frac{\varepsilon^2}{2}\left(\left\langle \xi, Q \xi \right\rangle + o(1)\right)\right) = \pi_{\lfloor t/\varepsilon\rfloor}^{(d/2,0)}\left(1-\frac{t^2\left\langle \xi, Q \xi \right\rangle + o(1)}{2 \lfloor t/\varepsilon\rfloor^2}\right) \\
	&\xrightarrow[\varepsilon \to 0]{} 2^{d/2} \Gamma\left(\frac{d}{2}+1\right) \langle \xi, Q \xi \rangle^{-d/4} t^{-d/2} J_{d/2}\left(t \langle \xi, Q \xi \rangle^{1/2}\right) = \hat{u}(t,\xi) = \hat{\mu}_t(\xi)\, ,
	\end{align*}
	where we used the formula for $\hat{u}(t,\xi)$ from Proposition~\ref{prop:fourier-fundamental-solution}. This finishes the proof of the weak convergence.
\end{proof}

\subsection{Proof of Local Convergence with Sinc Filter}



\begin{proof}[Proof of Theorem~\ref{thm:local-jpi}]
	This proof is similar to the one of Theorem~\ref{thm:weak-epd}. By Plancherel's formula, 
	\begin{equation}
	\label{eq:aux5-2}
	\sum_{v \in \Z^d} \left( x_n(v) - \left(u(n,.) * \psi\right)(v)\right)^2 
	= \left\Vert x_n - u(n,.)*\psi\right\Vert^2_{\ell^2(\Z^d)} 
	= \frac{1}{(2\pi)^d}\left\Vert \hat{x}_n - \hat{u(n,.)*\psi}\right\Vert^2_{L^2([-\pi,\pi]^d)} \,. 
	\end{equation}
	In this last expression, we take the Fourier transform of $u(n,.)*\psi$ as a function of $v \in \Z^d$. However, to decompose the computation, let us first compute the Fourier transform of $y \in \R^d \mapsto (u(n,.)*\psi)(y)$. The Fourier product of this convolution is the product of the Fourier transforms, and $\psi$ is chosen specifically so that its Fourier transform is $\hat{\psi}(\xi) = \bfone_{\{\xi \in [-\pi, \pi]^d\}}$. As a consequence, the Fourier transform of $y \in \R^d \mapsto (u(n,.)*\psi)(y)$ is $\xi \in \R^d \mapsto \hat{u}(t,\xi) \bfone_{\{\xi \in [-\pi, \pi]^d\}}$.
	
	We now discretize this function and seek the Fourier transform of $v \in \Z^d \mapsto (u(n,.)*\psi)(v)$. The Fourier transform of the discretization is the periodization of the Fourier transform \cite[Theorem 3.1]{mallat1999wavelet}, thus the Fourier transform of $v \in \Z^d \mapsto (u(n,.)*\psi)(v)$ is $\xi \in [-\pi, \pi]^d \mapsto \hat{u}(n,\xi)$. 
	
	We obtain 
	\begin{align*}
	\left\Vert \hat{x}_n - \hat{u(n,.)*\psi}\right\Vert^2_{L^2([-\pi,\pi]^d)} = \int_{[-\pi,\pi]^d} \diff\xi \left(\pi_n^{(d/2,0)}\left(\hat{\omega}(\xi)\right) - \hat{u}(n,\xi)\right)^2 \, .
	\end{align*}
	We make the change of variables $\zeta = n\xi$:
	\begin{align}
	\label{eq:aux5-3}
	\left\Vert \hat{x}_n - \hat{u(n,.)*\psi}\right\Vert^2_{L^2([-\pi,\pi]^d)} = n^{-d}\int_{\R^d} \diff\zeta \left(\pi_n^{(d/2,0)}\left(\hat{\omega}\left(\frac{\zeta}{n}\right)\right) - \hat{u}\left(n,\frac{\zeta}{n}\right)\right)^2 \bfone_{\{\zeta \in [-n\pi,n\pi]^d\}} \, .
	\end{align}
	Fix $\zeta \in \R^d$. Using the Mehler--Heine asymptotic (Proposition~\ref{prop:jacobi}.\eqref{it:mehler-heine}), we prove that 
	\begin{align*}
	\pi_n^{(d/2,0)}\left(\hat{\omega}\left(\frac{\zeta}{n}\right)\right) - \hat{u}\left(n,\frac{\zeta}{n}\right) \xrightarrow[n\to\infty]{} 0 \, .
	\end{align*}
	We do not repeat the computations because they are similar to the pointwise convergence in the proof of Theorem~\ref{thm:weak-epd}. This proves that the integrand of~\eqref{eq:aux5-3} converges pointwise to $0$. We want to apply the dominated convergence theorem to conclude, and thus seek a domination of 
	\begin{align}
	&\left(\pi_n^{(d/2,0)}\left(\hat{\omega}\left(\frac{\zeta}{n}\right)\right) - \hat{u}\left(n,\frac{\zeta}{n}\right)\right)^2 \bfone_{\{\zeta \in [-n\pi,n\pi]^d\}} \\ 
	&\qquad\leq 2\pi_n^{(d/2,0)}\left(\hat{\omega}\left(\frac{\zeta}{n}\right)\right)^2 \bfone_{\{\zeta \in [-n\pi,n\pi]^d\}} + 2 \hat{u}\left(n,\frac{\zeta}{n}\right)^2 \, . \label{eq:aux5-4}
	\end{align}
	By scale invariance of the EPD equation (or, more simply, from Proposition~\ref{prop:fourier-fundamental-solution}), $\hat{u}\left(n,\frac{\zeta}{n}\right) = \hat{u}\left(1,{\zeta}\right)$. Further, by Plancherel's theorem, 
	\begin{align*}
	\int_{\R^d} \diff\zeta \, \hat{u}(1,\zeta)^2 = (2\pi)^d \int_{\R^d} \diff y \, u(1,y)^2 < \infty \, ,
	\end{align*}
	thus the second term of~\eqref{eq:aux5-4} is independent of $n$ and integrable. We now need to find a domination for the first term. Here, the reasoning is similar to the boundedness of $\mu_{t, \varepsilon}$ in the proof of Theorem~\ref{thm:weak-epd}.
	\begin{itemize}
		\item If $\Vert \zeta \Vert \geq \frac{1}{\sqrt{\lambda}}$, by Lemma~\ref{lem:curien}, $\left\vert \hat{\omega}\left(\frac{\zeta}{n}\right)\right\vert \leq 1 - \frac{1}{n^2}$, thus by Proposition~\ref{prop:jacobi}.\eqref{it:bound-jacobi}, 
		\begin{align*}
		\pi_n^{(d/2,0)}\left(\hat{\omega}\left(\frac{\zeta}{n}\right)\right)^2 &\leq C_1^2 \left(\arccos\left\vert \hat{\omega}\left(\frac{\zeta}{n}\right)\right\vert\right)^{-d-1} n^{-d-1} \\
		&\leq C_1^2 \left(\arccos\left(1-\lambda \frac{\Vert\zeta\Vert^2}{n^2}\right)\right)^{-d-1} n^{-d-1} \, .
		\end{align*}
		There exists $C_9 > 0$ such that $\arccos(1-z) \geq C_9 \sqrt{z}$. Thus 
		\begin{align*}
		\pi_n^{(d/2,0)}\left(\hat{\omega}\left(\frac{\zeta}{n}\right)\right)^2 &\leq C_{10} \Vert \zeta \Vert^{-d-1} \, .
		\end{align*}
		\item If $\Vert \zeta \Vert < \frac{1}{\sqrt{\lambda}}$, then 
		\begin{align*}
		\pi_n^{(d/2,0)}\left(\hat{\omega}\left(\frac{\zeta}{n}\right)\right)^2 &\leq C_2^2 \, .
		\end{align*}
			\end{itemize} 
		We thus define the domination 
		\begin{align*}
		g(\zeta) = \begin{cases}
		C_{10} \Vert \zeta \Vert^{-d-1} &\text{if }\Vert \zeta \Vert \geq \frac{1}{\sqrt{\lambda}}, \\
		C_2^2 &\text{if }\Vert \zeta \Vert < \frac{1}{\sqrt{\lambda}}.
		\end{cases}
		\end{align*}
		This domination is integrable on $\R^d$; this concludes the theorem. 
\end{proof}

\section{Proof of Corollary~\ref{coro:sharp-rates}}
	Note that $\sum_{v \in \Z^d} x_n(v)^2 = \Vert x_n \Vert^2_{l^2(\Z^d)}$ and by \cref{thm:local-jpi}, 
	\begin{align*}
	\left\vert  \left\Vert x_n \right\Vert_{l^2(\Z^d)} - \left\Vert u(n,.) * \psi \right\Vert_{l^2(\Z^d)} \right\vert \leq \left\Vert x_n - u(n,.) * \psi\right\Vert_{l^2(\Z^d)} = o(n^{-d/2}) \, .
	\end{align*}
	It is thus sufficient to prove that 
	\begin{align*}
	\left\Vert u(n,.) * \psi \right\Vert^2_{l^2(\Z^d)}\sim \frac{1}{(\det Q)^{1/2} |B(0,1)| }\frac{1}{n^d} \, .
	\end{align*}
	In the proof of \cref{thm:local-jpi}, we explain that the Fourier transform of $v \in \Z^d \mapsto (u(n,.)*\psi)(v)$ is $\xi \in [-\pi, \pi]^d \mapsto \hat{u}(n,\xi)$. Thus by Plancherel's theorem,
	\begin{align*}
	\left\Vert u(n,.) * \psi \right\Vert^2_{l^2(\Z^d)} =  \frac{1}{(2\pi)^d} \left\Vert \hat{u}(n,.)\right\Vert^2_{L^2([-\pi,\pi]^d)} = \frac{1}{(2\pi)^d}\int_{[-\pi,\pi]^d} \diff \xi \, \hat{u}(n,\xi)^2 \, .
	\end{align*} 
	We make the change of variables $\zeta = \xi/n$. Note that by scale invariance of the EPD equation (or, more simply, from Proposition~\ref{prop:fourier-fundamental-solution}), $\hat{u}\left(n,\frac{\zeta}{n}\right) = \hat{u}\left(1,{\zeta}\right)$. Thus 
	\begin{align}
	\label{eq:aux5-5}
	\left\Vert u(n,.) * \psi \right\Vert^2_{l^2(\Z^d)} =  \frac{1}{(2\pi)^dn^d} \int_{[-n\pi,n\pi]^d} \diff \zeta \, \hat{u}(1,\zeta)^2 = \frac{1}{(2\pi)^dn^d} \left(\int_{\R^d} \diff \zeta \, \hat{u}(1,\zeta)^2 +o(1) \right) \, .
	\end{align}
	We use again Plancherel's theorem and then~\eqref{eq:fundamental-solution-d/2}: 
	\begin{align*}
	\frac{1}{(2\pi)^d} \int_{\R^d} \diff \zeta \, \hat{u}(1,\zeta)^2 = \int_{\R^d} \diff y \, u(1,y)^2 =  \left(\frac{\Gamma(d/2+1)}{\pi^{d/2}(\det Q)^{1/2}}\right)^2 \left\vert \left\{y \middle\vert \left\langle y , Q^{-1} y \right\rangle \leq 1 \right\}\right\vert \, .
	\end{align*}
	As the volume of the $d$-dimensional unit ball is $|B(0,1)| = \frac{\pi^{d/2}}{\Gamma\left(d/2+1\right)}$, the volume of the ellipsoid $\left\vert \left\{y \middle\vert \left\langle y , Q^{-1} y \right\rangle \leq 1 \right\}\right\vert$ is $\frac{\pi^{d/2}(\det Q)^{1/2}}{\Gamma\left(d/2+1\right)}$, thus 
	\begin{align*}
	\frac{1}{(2\pi)^d} \int_{\R^d} \diff \zeta \, \hat{u}(1,\zeta)^2 =\frac{\Gamma(d/2+1)}{\pi^{d/2}(\det Q)^{1/2}} = \frac{1}{(\det Q)^{1/2} |B(0,1)|} \, .
	\end{align*}
	Substituting in~\eqref{eq:aux5-5}, this concludes the proof.

\end{document}